%% file: main.tex
\documentclass[11pt,a4paper]{article}

\input{mymacros}

\begin{document}

\title{Tensor decomposition beyond uniqueness, with an application to the minrank problem}

\author{Pascal Koiran\thanks{Univ Lyon, EnsL, UCBL, CNRS,  LIP, F-69342, LYON Cedex 07, France. 
\email{pascal.koiran@ens-lyon.fr}.}
\and
Rafael Oliveira\thanks{University of Waterloo. \ \email{rafael@uwaterloo.ca}}
}

\maketitle

\begin{abstract}
    We prove a generalization to Jennrich's uniqueness theorem for tensor decompositions in the undercomplete setting.
    Our uniqueness theorem is based on an alternative definition of the standard tensor decomposition, which we call matrix-vector decomposition.
    Moreover, in the same settings in which our uniqueness theorem applies, we also design and analyze an efficient randomized algorithm to compute the unique minimum matrix-vector decomposition (and thus a tensor rank decomposition of minimum rank). 
    As an application of our uniqueness theorem and our efficient algorithm, we show how to compute all matrices of minimum rank (up to scalar multiples) in certain generic vector spaces of matrices.
\end{abstract}

%\linenumbers

%==================================================
\section{Introduction}
%==================================================

\input{introduction}

%==================================================
\section{Preliminaries}\label{section: preliminaries}
%==================================================

\input{preliminaries}

%==================================================
\section{Uniqueness Theorems} \label{sec:unique}
%==================================================

After some preliminaries, we establish the equivalence between \cref{th:unique1,th:unique2} in Section~\ref{sec:state}. 
We then prove \cref{th:unique2} in Section~\ref{sec:proof}.
The proofs are based on the notion of (generalized) eigenvalues and eigenvectors for a matrix pencil, which can be found in Section~\ref{sec:eigen}.

%==================================================
\subsection{Equivalence of the uniqueness theorems}\label{sec:state}
%==================================================

Recall that a tensor $T \in \kk^{m \times n \times p}$ can be cut into $p$ "slices" $Z_1,\ldots,Z_p$ where each slice is a $m \times n$ matrix. 
These are the 3-slices of $T$.
One can also cut $T$ in the two other directions into its 1-slices and 2-slices. 
In this paper we will only work with the $3$-slices, and henceforth the term ``slices'' will refer to the 3-slices. 
It follows immediately from \cref{def:matrix_vector} that the slices are linear combinations of the matrices $M_{\ell}$ occuring in a decomposition of $T$, namely,
\begin{equation} \label{eq:slicematrix}
Z_k = \sum_{\ell =1}^q w_{\ell k}M_{\ell}.
\end{equation}
Moreover, it is well known that each slice can be expressed as a product of three matrices defined from the vectors occurring in a decomposition of $T$ in its traditional form~(\ref{eq:decomp}). 
Namely, if $\rk T = r$, we have
\begin{equation} \label{eq:slice}
Z_k = U D_k V^T,
\end{equation}
%where $U$ is the $n \times r$ !!typo
where $U$ is the $m \times r$ matrix having the $u_i$ as column vectors,  $V$ is the $r \times n$ matrix having the $v_i$ as column vectors, and $D_k = \diag(w_{1k},\ldots,w_{rk})$. 
% \RO{notation overload here? These are not the same as the $w_{k \ell}$ from \eqref{eq:slicematrix}}
This notation is consistent with e.g.~\cite{moitra18}; note however that the transposed notation $Z_k = U ^TD_k V$ was used
in~\cite{koi24overcomplete} 
(i.e., $U$ and $V$ were defined as the matrices having the $u_i$ and $v_i$ as {\em row vectors}).
We record the following simple consequence of~(\ref{eq:slicematrix}) and~(\ref{eq:slice}):
\begin{lemma} \label{lem:maxrank}
Any matrix in the span of the slices of $T$ is of rank at most $\rk(T)$.
\end{lemma}
\begin{proof}
Let $r=\rk(T)$. By~(\ref{eq:slice}), any matrix $Z$ in the span of the $Z_k$ is of the form $Z = UDV^T$, where $D$ is a linear combination of the $D_k$. Since $D \in M_r(\mathbb{K})$, we have the upper bound $\rk(Z) \leq r$.

This can also be seen from a matrix-vector decomposition of $T$. 
Indeed, $T$ has a matrix-vector decomposition of rank $r$ by Proposition~\ref{prop:smallest}.
By~(\ref{eq:slicematrix}), $Z$ is a linear combination of the $M_{\ell}$ hence $\rk(Z) \leq \sum_{\ell} \rk(M_{\ell}) =r$.
\end{proof}

The following simple lemma, follows directly from~(\ref{eq:slicematrix})
\begin{lemma} \label{lem:mult}
Let $T \in \mathbb{K}^{m \times n \times p}$ be a tensor with slices $T_1,\ldots,T_p$ and a decomposition 
$T=\sum_{\ell=1}^q M_{\ell} \otimes w_{\ell}$. For any $A \in M_m(\mathbb{K})$, the tensor $T'$ with slices $AT_1,\ldots,AT_p$ 
admits the decomposition $T'=\sum_{\ell=1}^q (AM_{\ell}) \otimes w_{\ell}$.
\end{lemma}

% \begin{theorem}[first form of the uniqueness theorem] \label{th:unique1}
% Let $T=\sum_{i=1}^r u_i \otimes v_i \otimes w_i$ be a tensor of format $m \times n \times p$ such that:
% \begin{itemize}
% \item[(i)] The vectors $u_i$ are linearly independent.
% \item[(ii)] The vectors $v_i$ are linearly independent.
% \item[(iii)] Every vector $w_i$ is nonzero.
% \end{itemize}
% Then $\rk(T)=r$, and $T$ has %an essentially
% a  unique matrix-vector decomposition of rank~$r$.
% \end{theorem}

% \thmUniqueA*

% \begin{theorem}[second form of the uniqueness theorem] \label{th:unique2}
% Suppose that a tensor $T$  of format $m \times n \times p$ has a matrix-vector decomposition of the form:
% \begin{equation} \label{eq:uniqueth}
% T=\sum_{\ell=1}^q M_{\ell} \otimes w_{\ell}
% \end{equation}
% where the linear spaces $\Ima(M_{\ell})$ are in direct sum  { and where the linear spaces $\Ima(M_{\ell}^T)$ are also in direct sum.} Then $\rk(T)=\sum_{\ell = 1}^q \rk(M_{\ell})$, and (\ref{eq:uniqueth})
% is the unique matrix-vector decomposition of $T$ of minimal rank.
% \end{theorem}

\begin{lemma} \label{lem:rankr}
Let $\mathbb{K}$ be an arbitrary field. Let $(u_1,\ldots,u_r)$ and $(v_1,\ldots, v_r)$ be two families of  vectors of $\mathbb{K}^n$, respectively of rank $r_u$ and $r_v$. 
For the  matrix $M=\sum_{i=1}^r  u_i v_i^T$ we have $\rk M \leq \min(r_u,r_v)$. Moreover, if $r_u=r_v=r$ then $\rk M = r$ as well.
\end{lemma}

\begin{proof}
In the expression for $M$ we can rewrite each $u_i$ as a linear combination of the elements of a basis $e_1,\ldots,e_{r_u}$. 
This yields an expression for $M$ as a sum of $r_u$ matrices of rank at most 1, namely, 
$$M=\sum_{i=1}^{r_u}  e_i w_i^T$$
where the $w_i$ are linear combinations of the $v_i$.
 Hence $\rk M \leq r_u$, and $\rk M \leq r_v$ by a similar argument.

Assume now that $r_u=r_v=r$. It remains to show that $\rk M = r$. This is equivalent to $\dim \ker M = n-r$. 
A vector $x \in \mathbb{K}^n$ is in the kernel
 if and only if $\sum_{i=1}^{r} (v_i^T x) u_i =0$. Since the $u_i$ are linearly independent, this is equivalent
 to $v_i^T x =0$ for all $i$. Using now the linear independence of the $v_i$, it follows that the solution space is of dimension $n-r$ as needed.
\end{proof}

\noindent We are now ready to show the equivalence of the two uniqueness theorems.

\begin{proposition} \label{prop:form1to2}
Let $T=\sum_{i=1}^r u_i \otimes v_i \otimes w_i$ be a tensor of format $m \times n \times p$ such that:
\begin{itemize}
\item[(i)] The vectors $u_i$ are linearly independent.
\item[(ii)] The vectors $v_i$ are linearly independent.
\item[(iii)] Every vector $w_i$ is nonzero.
\end{itemize}
Then $T$ has a matrix-vector decomposition 
$T=\sum_{\ell=1}^q M_{\ell} \otimes w'_{\ell}$ where  the linear spaces $\Ima(M_{\ell})$ are in direct sum,  
the linear spaces $\Ima(M_{\ell}^T)$ are in direct sum, 
$r=\sum_{\ell = 1}^q \rk(M_{\ell})$ and $\{w'_1,\ldots,w'_{\ell}\} \subseteq \{w_1,\ldots,w_r\}$.
\end{proposition}
\begin{proof}
Let us group together the $w_{\ell}$ that are colinear, like in the proof of Proposition~\ref{prop:smallest}. 
This yields a matrix-vector decomposition $T=\sum_{\ell=1}^q M_{\ell} \otimes w'_{\ell}$ where
 $\{w'_1,\ldots,w'_{\ell}\} \subseteq \{w_1,\ldots,w_r\}$. By Lemma~\ref{lem:rankr}, the matrix $M=\sum_{\ell = 1}^q M_{\ell}$
 is of rank $r$. This implies that $r=\sum_{\ell = 1}^q \rk(M_{\ell})$ 
 and that the  the linear spaces $\Ima(M_{\ell})$ are in direct sum.
  A similar reasoning for  $M^T=\sum_{\ell = 1}^q M_{\ell}^T$ shows that the linear spaces $\Ima(M_{\ell}^T)$ are in direct sum.
 \end{proof}
Here is a converse to this proposition.
\begin{proposition} \label{prop:form2to1}
Suppose that a tensor $T$  of format $m \times n \times p$ has a matrix-vector decomposition of the form:
$$T=\sum_{\ell=1}^q M_{\ell} \otimes w'_{\ell}$$
where the linear spaces $\Ima(M_{\ell})$ are in direct sum, 
and the linear spaces $\Ima(M_{\ell}^T)$ are also in direct sum. 
 Then $T$ has a decomposition of the form 
$T=\sum_{i=1}^r u_i \otimes v_i \otimes w_i$ where $r=\sum_{\ell = 1}^q \rk(M_{\ell})$ and:
\begin{itemize}
\item[(i)] The vectors $u_i$ are linearly independent.
\item[(ii)] The vectors $v_i$ are linearly independent.
\item[(iii)] Every vector $w_i$ is nonzero.
\end{itemize}
Moreover, we can assume that $w_i \in \{w'_1,\ldots,w'_{\ell}\}$ for all $i$.
\end{proposition}
\begin{proof}
We apply the transformation of Proposition~\ref{prop:form1to2} in reverse, i.e., we write each  $M_{\ell}$ as a sum of rank 1 matrices:
 $$M_{\ell} = \sum_{k=1}^{\rk(M_{\ell})} u_{k\ell} \otimes v_{k\ell}.$$
 Then we obtain the desired decomposition of $T$ by expanding each product  $M_{\ell} \otimes w'_{\ell}$.
 In the resulting expansion, the $u_i$ are linearly independent since they form a basis of $\bigoplus_{\ell} \Ima(M_{\ell})$.
A similar reasoning for the $M_{\ell}^T$ shows that the $v_i$ are linearly independent.
The $w_i$ are nonzero since they are the same vectors as the $w'_i$ (each $w'_i$ is repeated $\rk(M_i)$ times in the list of the $w_i$).
\end{proof}

%========================================================
\subsection{Proof of the uniqueness theorems}
\label{sec:proof}
%========================================================

In this section, we now prove \cref{th:unique2}. 
\emph{Throughout this section, we assume that $T$ is a tensor satisfying the (equivalent) hypotheses of} \cref{th:unique1,th:unique2}.
We also assume that our field $\mathbb{K}$ is infinite. 
This is without loss of generality since uniqueness of decomposition for a field implies uniqueness for all its subfields.
 
As a main step toward the uniqueness theorems, we show in Theorem~\ref{th:uniqueimage2} that the linear spaces $\Ima M_\ell$ are the same in all matrix-vector decompositions $T=\sum_{\ell =1}^q M_{\ell} \otimes \gamma_{\ell}$ of minimal rank.
The uniqueness of the matrix-vector decomposition then follows from a simple direct sum argument.

The next claim follows from a simple argument which appears in~\cite{koi24overcomplete} before Proposition~11.
We recall the proof here for the sake of completeness.

\begin{proposition} \label{prop:invert1}
The span of the slices of $T$ contains a matrix of rank $r$.
\end{proposition}
\begin{proof}
By~(\ref{eq:slice}), the matrices in $\langle Z_1,\ldots,Z_p \rangle$
are exactly the matrices of the form:
\begin{equation} \label{eq:span2}
Z=UDV^T,\ D = \sum_{k=1}^p c_k D_k
\end{equation}
where $c_1,\ldots,c_p \in \mathbb{K}$.  Each entry of $D$ is a linear form in the $c_k$.
Since $w_i \neq 0$, for all $i$ there exists $k$  such that $w_{ik} \neq 0$. As a result, these linear forms are all nonzero.  
Indeed, 
\begin{equation} \label{eq:diag}
D=\diag(\langle c,w_1 \rangle,\ldots, \langle c,w_r \rangle)
\end{equation}
and
the $k$-th coefficient of the $i$-th linear form is $(D_k)_{ii}=w_{ik} \neq 0$.
{Since $\mathbb{K}$ is infinite}, there exists $c_1,\ldots,c_p \in \mathbb{K}$ such that the matrix $D$ in~(\ref{eq:span2}) 
is invertible, and the corresponding $Z=UDV^T$ is of rank~$r$ since $U,D$ and $V^T$ are all of full rank $r$.
\end{proof}
It follows from \cref{prop:invert1} and \cref{lem:maxrank} that $\rk(T)=r$ as claimed in \cref{th:unique1}.
\begin{corollary} \label{cor:direct2}
In any matrix-vector decomposition of $T$ of rank $r$, the linear spaces $\Ima(M_{\ell})$ are in direct sum.
\end{corollary}
\begin{proof}
Let $T=\sum_{\ell =1}^q M_{\ell} \otimes w_{\ell}$ be a matrix-vector decomposition of $T$ of rank~$r$. 
% We use the notation $\gamma_{\ell}$ here because we don't know (yet) that these are the same vectors as the $w_i$ in the statement of \cref{th:unique1}.
By \cref{prop:invert1} there is a matrix $Z$ of rank $r$ in the span of the slices of~$T$, and by~\cref{eq:slicematrix} $Z$ is a linear combination of the $M_{\ell}$. 
Therefore,
$$r=\dim \Ima(Z) \leq \dim (\sum_{\ell} \Ima(M_{\ell})) \leq \sum_{\ell} \dim (\Ima M_{\ell})=r.$$
Here the  last equality follows from the definition of the rank of a matrix-vector decomposition.
We conclude that
$$\dim ( \sum_{\ell} \Ima M_{\ell}) = \sum_{\ell} \dim  (\Ima M_{\ell}),$$
and the spaces $\Ima(M_{\ell})$ are indeed in direct sum.
\end{proof}
\begin{corollary} \label{cor:direct3}
In any matrix-vector decomposition of $T$ of rank $r$, the linear spaces $\Ima(M_{\ell}^T)$ are in direct sum.
\end{corollary}
\begin{proof}
Let $T=\sum_{\ell =1}^q M_{\ell} \otimes \gamma_{\ell}$
 be a matrix-vector decomposition of $T$ of rank~$r$. 
 Consider the tensor $T'=\sum_{i=1}^r v_i \otimes u_i \otimes w_i$, obtained from $T$ by exchanging $u_i$ and $v_i$
 in the decomposition $T=\sum_{i=1}^r u_i \otimes v_i \otimes w_i$.
 For $T'$ we have the matrix-vector decomposition of rank $r$:
$$T'=\sum_{\ell =1}^q M_{\ell}^T \otimes \gamma_{\ell}.$$
The result therefore follows from \cref{cor:direct2} applied to $T'$.
\end{proof}
It remains to show that $T$ has a unique matrix-vector decomposition of rank~$r$.
First, we'll show (in \cref{th:uniqueimage2}) that the linear spaces $\Ima(M_{\ell})$ are the same in all matrix-vector decompositions of $T$. 

Henceforth, for any $c \in \mathbb{K}^p$ we denote by $T_c$ the linear combination $\sum_{k=1}^p c_kZ_k$, and likewise $D_c := \sum_{k=1}^p c_kD_k$ in~\cref{eq:diag}.
\begin{corollary} \label{cor:rank}
For a generic choice of $a,b$ in $\mathbb{K}^p$, $T_a$ and $T_b$ are of rank~$r$; the kernel of $T_a$ (and $T_b$) is of dimension $n-r$, and is equal to $\displaystyle \bigcap_{\ell \in [q]} \ker M_{\ell}$.
\end{corollary}

\begin{proof}
The first property follows from the proof of \cref{prop:invert1}: $T_a$ and $T_b$ are of rank $\leq r$ for any choice of $a$ and $b$, and the ranks will be equal to $r$ if $a$ and $b$ avoid a union of $r$ hyperplanes. 
The kernels are therefore of dimension $n-r$.

By \cref{eq:slicematrix}, $T_a = \sum_{\ell} \langle a,\gamma_{\ell} \rangle M_{\ell}$. 
If $M_{\ell} x =0$ for all $\ell$, this immediately shows that $T_a x =0$. 
The converse follows from the fact that $\Ima(M_{\ell})$ are in direct sum, and for a generic choice of $a$ we have  $\langle a,\gamma_{\ell} \rangle \neq 0$ for all $\ell$.
\end{proof}

\newpage

\begin{lemma} \label{lem:image2}
Let $T=\sum_{\ell =1}^q M_{\ell} \otimes \gamma_{\ell}$ be a matrix-vector decomposition of $T$ of rank $r$. 
The following properties hold for a generic choice of $a,b \in \mathbb{K}^p$:
\begin{itemize}
\item[(i)] If $x$ is an eigenvector of the pair $(T_a,T_b)$ then there is $\ell \in [q]$ such that $T_bx \in \Ima M_{\ell}$.

\item[(ii)] There are exactly $q$ nontrivial eigenvalues for   $(T_a,T_b)$, and if $\lambda$ is a nontrivial eigenvalue then $\langle a,\gamma_{\ell} \rangle = \lambda \langle b,\gamma_{\ell} \rangle$ for some $\ell \in [q]$. The dimension of the corresponding eigenspace $V_{\lambda}$ is equal to $n-r+\rk(M_{\ell})$, and the image of $V_{\lambda}$ by $T_b$ is equal to $\Ima M_{\ell}$.

\item[(iii)] If $x$ is an   eigenvector associated to the nontrivial eigenvalue 
$\lambda=\langle a,\gamma_{\ell} \rangle /  \langle b,\gamma_{\ell} \rangle$, we have $T_a x = \langle a,\gamma_{\ell} \rangle M_{\ell}x$ and 
$T_b x =  \langle b,\gamma_{\ell} \rangle M_{\ell}x$.
\end{itemize}
\end{lemma}

\begin{proof}
Let $x$ be an eigenvector of $(T_a, T_b)$ associated to the eigenvalue $\lambda$.
If $T_b x = 0$ then certainly $T_bx \in \Ima M_{\ell}$ for any $\ell$. 
For the remainder of the proof of (i) we assume that $T_b x \neq 0$.
By \cref{eq:slicematrix}, $T_a = \sum_{\ell} \langle a,\gamma_{\ell} \rangle M_{\ell}$ and 
$T_b = \sum_{\ell} \langle b,\gamma_{\ell} \rangle M_{\ell}$. This leads to
\begin{equation} \label{eq:2decomp2}
 \sum_{\ell} \langle a,\gamma_{\ell} \rangle M_{\ell}x= \sum_{\ell} \lambda \langle b,\gamma_{\ell} \rangle M_{\ell}x.
 \end{equation}
  
Remember from \cref{cor:direct2} that the spaces $\Ima M_{\ell}$ are in direct sum. 
As a result, we must have  $\langle a,\gamma_{\ell} \rangle = \lambda \langle b,\gamma_{\ell} \rangle$ for all the $\ell$ 
such that $M_{\ell}x \neq 0$ in \cref{eq:2decomp2}. We claim that there is in fact exactly one $\ell$ such that $M_{\ell}x \neq 0$.
Note first that one cannot have $M_{\ell} x =0$ for all $\ell$ since $T_bx \neq 0$.
   To continue the proof of the claim, let us assume that
 $$\langle a,\gamma_{\ell} \rangle = \lambda \langle b,\gamma_{\ell} \rangle\ \mathrm{and}\ \langle a,\gamma_{\ell'} \rangle = \lambda \langle b,\gamma_{\ell'} \rangle$$
 for two distinct indices $\ell \neq \ell'$. 
  This implies that 
 \begin{equation} \label{eq:bilinear}
 \langle a,\gamma_{\ell} \rangle \langle b,\gamma_{\ell'} \rangle 
 -  \langle a,\gamma_{\ell'} \rangle  \langle b,\gamma_{\ell} \rangle = 0.
 \end{equation}
 It remains to show that this bilinear form in $a$ and $b$ is not identically 0: this will show that for a generic choice of $a$ and $b$, (\ref{eq:bilinear}) cannot hold. The matrix of this bilinear form is equal to 
 $\gamma_{\ell} \gamma_{\ell'}^T -  \gamma_{\ell'} \gamma_{\ell}^T$. It is indeed nonzero
  since $\gamma_{\ell}$ and 
 $\gamma_{\ell'}$ are not colinear by definition of a matrix-vector decomposition. This completes the proof of the claim.
 We have therefore shown that  $T_a x = \langle a,\gamma_{\ell} \rangle M_{\ell}x$
 and $T_b x = \langle b,\gamma_{\ell} \rangle M_{\ell}x$, i.e., $T_a x$ and $T_bx$ both belong to $\Ima M_{\ell}$.
 This completes the proof of (i) and (iii); we have also shown that if  $\lambda$ is a nontrivial eigenvalue then $\lambda = \langle a,\gamma_{\ell} \rangle /  \langle b,\gamma_{\ell} \rangle$ for some $\ell \in [q]$ (note that $ \langle b,\gamma_{\ell} \rangle$ is generically nonzero for all $\ell$).
Next, we compute the dimension of the corresponding eigenspace $V_{\lambda}$.

Suppose for instance that $\lambda=\langle a,\gamma_{q} \rangle /  \langle b,\gamma_{q} \rangle$. 
The proof of (i) shows that $x \in V_{\lambda}$ if and only if $M_{\ell} x =0$ for all $\ell \neq q$. Let $M$ be the matrix of this linear system.
We have the block decomposition $M^T=(M_1^T \ldots  M_{q-1}^T)$. 
By \cref{cor:direct3} this implies that 
$$\rk(M)=\rk(M^T)=\sum_{\ell =1}^{q-1} \rk(M_{\ell}^T) = r - \rk(M_{q}),$$
and this shows that $\dim V_{\lambda} = n-r+\rk(M_{q})$ as claimed.

Finally, we build on this computation of $\dim V_{\lambda}$ to show that $T_b (V_{\lambda}) = \Ima M_{\ell}$. 
From (iii) we already know that  $T_b (V_{\lambda}) \subseteq \Ima M_{\ell}$. In order to show that these two subspaces 
are equal we will show that they have the same dimension. 
For this, we'll compute the dimension of the kernel
$K_{\lambda}$ of $T_b$
viewed as a linear operator on $V_{\lambda}$. Let us continue to assume for instance that $\lambda=\langle a,\gamma_{q} \rangle /  \langle b,\gamma_{q} \rangle$. Recall that  $x \in V_{\lambda}$ if and only if $M_{\ell} x =0$ for all $\ell \neq q$.
Therefore, $x \in K_{\lambda}$ if and only if $M_{\ell} x =0$ for all $\ell \in [q]$. This shows that $K_{\lambda}$ is actually
independent of $\lambda$, and is equal (by \cref{cor:rank}) to $\ker T_b$. Hence $\dim K_{\lambda} = n-r$,
and $\dim T_b (V_{\lambda}) = \dim V_{\lambda} - \dim K_{\lambda} = \rk(M_q)$.
\end{proof}

\begin{theorem} \label{th:uniqueimage2}
Let $T=\sum_{\ell =1}^q M_{\ell} \otimes \gamma_{\ell}$ be a matrix-vector decomposition %of $T$ 
of rank $r$. 
If $T=\sum_{\ell =1}^{q'} M'_{\ell} \otimes \gamma'_{\ell}$ is another matrix-vector decomposition of rank~$r$, then $q=q'$ and there is a permutation $\pi$ such that $\Ima M'_{\ell} = \Ima M_{\pi(\ell)}$.
\end{theorem}
\begin{proof}
Since the intersection of two Zariski open sets is nonempty (and even Zariski open), there exist $a,b \in \mathbb{K}^p$ such that the 3 properties of \cref{lem:image2} apply to 
our two decompositions of $T$. In particular, by \cref{lem:image2}, $q$ and $q'$ are both equal to the number of nontrivial eigenvalues of the pair $(T_a,T_b)$. 
Moreover, the images of the corresponding eigenspaces are the spaces $\Ima M_{\ell}$, and (applying the lemma to the second decomposition of $T$) these images are {\em also} the spaces $\Ima M'_{\ell}$.
This shows that the second family of spaces is obtained from the first by a permutation of indices.
\end{proof}
With the above theorem at hand, the proof of \cref{th:unique2} then follows from two simple lemmas, which we now state and prove.

\begin{lemma} \label{lem:direct}
Suppose that  $E_1,\ldots,E_q \subseteq \mathbb{K}^m$ are in direct sum. The linear spaces of matrices:
$${\cal E}_1 = \{M \in M_{m,n}(\mathbb{K});\Ima(M) \subseteq E_1\},\ldots,{\cal E}_q = \{M \in M_{m,n}(\mathbb{K});\Ima(M) \subseteq E_q\}$$
are also in direct sum.
\end{lemma}
\begin{proof}
Suppose that $M_1+\ldots+M_q=0$ with $M_i \in {\cal E}_i$ for $i=1,\ldots,q$.
For any $x \in \mathbb{K}^n$, $M_1x  +\ldots+M_qx=0$ and $M_i x \in E_i$ for all $i$. This implies $M_ix=0$ for all $i$ since the $E_i$ are in direct sum. Since this is true for any $x$, we conclude that $M_i=0$ for all $i$.
\end{proof}
\begin{lemma} \label{lem:rank1}
If  the vectors  $u,u' \in \mathbb{K}^m$ and $v,v' \in \mathbb{K}^n$ satisfy $v_j u = v'_j u'$ for all $j=1,\ldots,n$ 
we have $u \otimes v = u' \otimes v'$.
\end{lemma}
\begin{proof}
As matrices, $u \otimes v$ and $u' \otimes v'$ are  represented respectively by $uv^T$ and $u'v'^T$. These two matrices  are equal since their respective columns are equal. 
\end{proof}

We can now complete the proof of the uniqueness theorems.
For convenience, we restate \cref{th:unique2} here.

\thmUniqueB*

\begin{proof}
Consider two matrix-vector decompositions of rank $r$:
$$T=\sum_{\ell =1}^q M_{\ell} \otimes \gamma_{\ell} =  \sum_{\ell =1}^{q'} M_{\ell'} \otimes \gamma_{\ell'}.$$
By \cref{th:uniqueimage2}, $q=q'$ and we can renumber the terms in this decomposition so that 
$\Ima M'_{\ell} = \Ima M_{\ell}$ for $\ell=1,\ldots,q$. 
Moreover, by \cref{eq:slicematrix} we have two expressions for each of the slices of $T$:
$$Z_k = \sum_{\ell =1}^q {\gamma}_{\ell k}M_{\ell} = \sum_{\ell =1}^q {\gamma'}_{\ell k}M'_{\ell}.$$
Since $\Ima M'_{\ell} = \Ima M_{\ell}$ and these spaces are in direct sum, \cref{lem:direct} 
shows that ${\gamma}_{\ell k}M_{\ell} = {\gamma'}_{\ell k}M'_{\ell}$ for all $\ell$. 
Since this applies to all $k$, we have $M_{\ell} \otimes {\gamma}_{\ell } = M'_{\ell} \otimes {\gamma'}_{\ell}$ by \cref{lem:rank1} and we have shown that $T$ has a unique matrix-vector decomposition of minimal rank.
\end{proof}

\begin{remark}\label{remark: adjoint algebra uniqueness}
There is an alternative way of proving the above uniqueness theorem via the uniqueness theorem for indecomposable vector space decompositions (the latter proved in \cite[Corollary B.2]{garg20}), which in turn uses the Krull-Schmidt theorem for modules.
In short, in the alternative approach, we can show that the direct sum conditions of the given matrix-vector decomposition imply that the adjoint algebra (see \cite[Appendix A]{garg20}) of the set of 3-slices of our tensor can be block-diagonalizable in a unique way.
With the above at hand, we can now prove an analogous statement to \cref{lem:image2,th:uniqueimage2} which show that any matrix-vector decomposition of minimum rank should have the same adjoint algebra as the 3-slices.
The above two facts imply uniqueness of matrix-vector decompositions.

It is important to notice that the uniqueness of vector space decompositions alone is not enough to prove our result, but that the above is an alternative way to interpret (and prove) our uniqueness theorem.
\end{remark}

%========================================================
\section{Decomposition algorithm} \label{sec:algo}
%========================================================

\input{decomposition-algorithm}

%========================================================
\section{The minimum rank problem}
\label{sec:rank}
%========================================================

\input{minrank}

{\small
\bibliographystyle{alpha}
%\bibliography{biblio}

}

\appendix

%=====================================================
\section{The simultaneous diagonalization algorithm}
\label{sec:simdiag}
%=====================================================

\input{simultaneous-diagonalization}

%=====================================================
\section{Computing  \texorpdfstring{$\Ima(M_{\ell})$}{ImaMell} with the Moore-Penrose inverse}\label{sec:images}
%=====================================================

\input{moore-penrose}

\end{document}

%% file: mymacros.tex
\usepackage{amsmath}
\usepackage[T1]{fontenc}
\usepackage[utf8]{inputenc}
\usepackage{amsfonts}
\usepackage{amsthm}
\usepackage{amstext}
\usepackage{amssymb}
\usepackage[noend]{algpseudocode}
\usepackage[ruled,vlined,linesnumbered]{algorithm2e}
\let\oldnl\nl% Store \nl in \oldnl
\newcommand{\nonl}{\renewcommand{\nl}{\let\nl\oldnl}}% Remove line number for one line
\usepackage{thm-restate}

\usepackage{xcolor}
\usepackage[bookmarksopen=true, pagebackref]{hyperref}
\hypersetup{
  colorlinks,
  linkcolor=[rgb]{0.3,0.3,0.6},
  citecolor=[rgb]{0.2, 0.6, 0.2},
  urlcolor=[rgb]{0.6, 0.2, 0.2}
}
\usepackage[nameinlink, capitalize]{cleveref}
\usepackage{lineno}
\usepackage{dirtytalk}

\newcommand{\email}[1]{Email: \href{mailto:#1}{\texttt{#1}}}

\ifdefined\DRAFT
  \usepackage[inline]{showlabels}
  \newcommand{\PK}[1]{{\color{blue}[Pascal: #1]}}
  \newcommand{\RO}[1]{{\color{red}[Rafael: #1]}}
\else
  \newcommand{\PK}[1]{}
  \newcommand{\RO}[1]{}
\fi

\newtheorem{theorem}{Theorem}[section]
\newtheorem*{theorem*}{Theorem}

\newtheorem{lemma}[theorem]{Lemma}

\newtheorem*{lemma*}{Lemma}
\newtheorem{proposition}[theorem]{Proposition}
\newtheorem{corollary}[theorem]{Corollary}
\newtheorem{remark}[theorem]{Remark}

\newtheorem*{open*}{Open~question}
\newtheorem{definition}[theorem]{Definition}

\newcommand{\cc}{\mathbb{C}}
\newcommand{\rr}{\mathbb{R}}

\newcommand{\nn}{\mathbb{N}}
\newcommand{\kk}{\mathbb{K}}

\newcommand{\diag}{\operatorname{diag}}
\newcommand{\rk}{\operatorname{rank}}

\newcommand{\Kspan}[1]{\operatorname{span}_\kk\{#1\}}

\DeclareMathOperator{\Ima}{Im}

%% file: introduction.tex
A tensor can be viewed as a multidimensional array with entries in some field~$\kk$. 
In this paper, we will only consider tensors of order 3, i.e., elements of $\kk^{m \times n \times p} = \kk^m \otimes \kk^n \otimes \kk^p$.
Given 3 vectors $u \in \kk^m, v \in \kk^n, w \in \kk^p$ we recall that their tensor product $u \otimes v \otimes w$ is the tensor $T \in \kk^{m \times n \times p}$ with entries: $T_{ijk}=u_i  v_j w_k$.
By definition, a tensor of this form with $u,v,w \neq 0$ is said to be of {\em rank one}.
The rank of an arbitrary tensor $T$ is defined as the smallest integer $r$ such that $T$ can be written as a sum of $r$ tensors of rank one (and the rank of $T=0$ is 0). 
The decomposition
\begin{equation} \label{eq:decomp} 
T=\sum_{i=1}^{r} u_i \otimes v_i \otimes w_i
\end{equation}
is said to be unique (or sometimes, "essentially unique") if up to a permutation, the rank-1 terms $u_i \otimes v_i \otimes w_i$ are the same in all decompositions of $T$ as a sum of $r$ tensors of rank one.

The starting point of this paper is a uniqueness theorem for decomposition of order 3 tensors and an associated decomposition algorithm:
\begin{theorem}[Jennrich's uniqueness theorem] \label{th:jennrich}
Let $T=\sum_{i=1}^r u_i \otimes v_i \otimes w_i$ be a tensor in $\kk^{m \times n \times p}$ such that:
\begin{itemize}
\item[(i)] The vectors $u_i$ are linearly independent.
\item[(ii)] The vectors $v_i$ are linearly independent.
\item[(iii)]  The vectors $w_i$ are pairwise independent.
\end{itemize}
Then $\rk(T)=r$, and the decomposition of $T$ as a sum of $r$ rank one tensors is essentially unique.
\end{theorem}
This result is a special case of Kruskal's uniqueness theorem~\cite{kruskal77}. 
In contrast to Kruskal's theorem, \cref{th:jennrich} has an efficient algorithmic proof.  
The resulting algorithm is known as the \say{simultaneous diagonalization} or \say{Jennrich's algorithm.}\footnote{The algorithm was not actually discovered by Robert Jennrich, so this name should be viewed as a tribute to his contributions to tensor decomposition (see~\cite{harshman70,harshman72}) rather than as a historically accurate attribution. 
The first published version of the algorithm seems to be from~\cite{leurgans93}.} The polynomial running time of this algorithm is a remarkable feature since tensor decomposition (or just computing the rank of a tensor) is in general NP-hard~\cite{hastad90}.
This property has made the simultaneous diagonalization algorithm a cornerstone of further algorithmic work on tensors.
For instance, this algorithm was used as a subroutine in~\cite{ma16} to obtain an algorithm for overcomplete decomposition of random tensors.  
The term {\em overcomplete} refers to the situation when the tensor rank is larger than the dimensions of the tensor; by contrast, in the {\em undercomplete} setting of \cref{th:jennrich} we must have $r \leq \min(m,n)$.
Some of the ideas behind Jennrich's algorithm have also inspired the spectral algorithm in~\cite{hopkins19}.
More recently, the simultaneous diagonalization algorithm was used in~\cite{koi24overcomplete} to obtain the first efficient algorithm for overcomplete decomposition of generic tensors of order 3. 
Jennrich's algorithm was also used in~\cite{johnston23} to find low rank matrices in matrix  subspaces. 
This result of~\cite{johnston23} was in turn used in~\cite{kothari24} to give another  efficient algorithm for overcomplete decomposition of generic tensors of order 3. 
In light of all these results, it is quite natural to look for more applications and generalizations of Jennrich's uniqueness theorem and of the corresponding decomposition algorithm.
This is the main focus of this paper.

%========================================================
\subsection{Our results}
%========================================================

In this paper we work in the undercomplete setting and prove three main results, which we now outline.
Our first result is %to prove 
a generalization of Jennrich's uniqueness theorem (\cref{th:jennrich}) where condition (iii) is removed.
Our second result is on the algorithmic side, where we also generalize Jennrich's decomposition algorithm to the setting of our uniqueness theorem.
As a consequence of these results, our third result is to obtain an efficient algorithm that finds all matrices of minimum rank %for 
in certain generically chosen subspaces of matrices.

To properly describe our results, we introduce some basic notation that will be used throughout the paper. 
We denote by $M_{m,n}(\kk)$ the set of matrices with $m$ rows, $n$ columns and entries in~$\kk$. 
We denote by $M_n(\kk)$ the set of square matrices of size $n$, by~$GL_n(\kk)$ the group of invertible matrices of size $n$,
and by $I_n$ the identity matrix of size $n$.

Before we formally state our results, we motivate the conceptual aspect behind our uniqueness theorem, which we call \emph{matrix-vector decompositions}.

It is well known that without condition (iii) in \cref{th:jennrich}, the tensor
decomposition of minimal rank is no longer unique.
One way in which uniqueness fails already happens in the decompositions of matrices: it is easy to see that the decomposition of a rank 2 matrix as the sum of two matrices of rank 1 is never unique; see \cref{lem:rank2} in \cref{sec:matrix} for a proof of this fact.
This failure of uniqueness may at first seem problematic, as most of the efficient tensor decomposition algorithms apply in a setting where the decomposition of smallest rank is known to be unique. 
Informally speaking, uniqueness (and the ingredients in a uniqueness proof) help a decomposition algorithm \say{zero in} on the correct decomposition.

The way we deal with the above obstacle to uniqueness is by considering a more relaxed decomposition of a tensor as a sum of tensor products of the form $M \otimes w$ where $M$ is a matrix and $w$ is a vector.
More precisely, we have the following definition:

\begin{definition}[Matrix-vector decompositions]\label{def:matrix_vector}
Let $T \in \mathbb{K}^{m \times n \times p}$ be a tensor. 
A matrix-vector decomposition of $T$ is a decomposition of the form
\begin{equation} \label{eq:matrixvector}
T=\sum_{\ell=1}^q M_{\ell} \otimes w_{\ell}
\end{equation}
where $M_{\ell} \in M_{m,n}(\mathbb{K}) \setminus \{0\}$, $w_{\ell} \in \mathbb{K}^p \setminus \{0\}$ for every $\ell = 1,\ldots,q$, and no two vectors~$w_{\ell}$ in this list are colinear.

The rank of this decomposition is defined as $\sum_{\ell =1}^q  \rk(M_{\ell})$.
\end{definition}

In \cref{prop:smallest} we show that the smallest rank of a matrix-vector decomposition equals the tensor rank of $T$.
Thus, the above decomposition generalizes the traditional tensor decomposition (the latter also requires that each $M_{\ell}$ be of rank 1), and it has the advantage that it avoids the aforementioned issue for matrices.  
In contrast to condition~(iii) in \cref{th:jennrich}, non-colinearity of the vectors $w_\ell$ can be assumed without loss of generality: if two vectors $w_{\ell}, w_{\ell'}$ are scalar multiples of each other, we can add up the corresponding (properly scaled) matrices $M_{\ell}, M_{\ell'}$ (and the number of terms in the decomposition goes from $q$ to $q-1$).

A related notion of decomposition is studied in~\cite{johnston23}.\footnote{An improved analysis of one of their results can be found in~\cite{dastidar25}.}
One key difference is that they want to minimize the number of terms in a decomposition, whereas we want to minimize the sum of the ranks of the $M_{\ell}$.

So far, it looks like all we've achieved with Definition~\ref{def:matrix_vector} is to redefine tensor rank in a slightly unusual way. 
But we will now see that this definition leads to new uniqueness results and decomposition algorithms.
This is the content of our main uniqueness theorem.

\begin{restatable}[Uniqueness theorem]{theorem}{thmUniqueA}\label{th:unique1}
Let $T=\sum_{i=1}^r u_i \otimes v_i \otimes w_i$ be a tensor in $\mathbb{K}^{m \times n \times p}$ such that:
\begin{itemize}
\item[(i)] The vectors $u_i$ are linearly independent.
\item[(ii)] The vectors $v_i$ are linearly independent.
\item[(iii)] Every vector $w_i$ is nonzero.
\end{itemize}
Then $\rk(T)=r$, and $T$ has 
a unique rank~$r$ matrix-vector decomposition.
\end{restatable}

We now give an equivalent version of the above theorem which is stated solely in terms of matrix-vector decompositions.
This version will be more convenient to use in certain parts of the paper.

\begin{restatable}[Uniqueness theorem, equivalent formulation]{theorem}{thmUniqueB}\label{th:unique2}
Suppose that a tensor $T \in \mathbb{K}^{m \times n \times p}$ has a matrix-vector decomposition of the form:
\begin{equation} \label{eq:uniqueth}
T=\sum_{\ell=1}^q M_{\ell} \otimes w_{\ell}
\end{equation}
where the linear spaces $\Ima(M_{\ell})$ are in direct sum  { and where the linear spaces $\Ima(M_{\ell}^T)$ are also in direct sum.} Then $\rk(T)=\sum_{\ell = 1}^q \rk(M_{\ell})$, and (\ref{eq:uniqueth})
is the unique matrix-vector decomposition of $T$ of minimum rank.
\end{restatable}

We prove the equivalence of these two theorems in \cref{sec:state}, and we prove the above theorems in \cref{sec:proof}.

With the above uniqueness theorems at hand, we are now ready to state our algorithmic contribution: under the uniqueness conditions, we can compute the minimum matrix-vector decomposition.

\newpage

\begin{restatable}[Matrix-vector decomposition algorithm]{theorem}{uniquenessAlg}\label{theorem: uniqueness algorithm}
Suppose that a tensor $T \in \mathbb{K}^{m \times n \times p}$ has a matrix-vector decomposition of the form:
$$T=\sum_{\ell=1}^q M_{\ell} \otimes w_{\ell}$$
where the linear spaces $\Ima(M_{\ell})$ are in direct sum  { and where the linear spaces $\Ima(M_{\ell}^T)$ are also in direct sum.} 

There is a randomized, polynomial-time algorithm (Algorithm~\ref{algo:main}) such that, on input $T$ as above, it outputs the above matrix-vector decomposition (as usual, the $M_{\ell}$ and $w_{\ell}$ are determined only up to scaling and permutation).
\end{restatable}

Note that our decomposition algorithm applies under the same conditions as in \cref{th:unique1} (or equivalently \cref{th:unique2}).

When $\kk=\mathbb{Q}$, our algorithm can be implemented in the Turing machine model of computation, and it runs in time polynomial in the bit size of the input tensor $T$ (Remark~\ref{rem:time}). 

In the case of general fields $\kk$, we assume that we have access to an algorithm for the computation of the roots of a univariate polynomial with coefficients in $\kk$ (for step 4 of Algorithm~\ref{algo:images}). 
This is a fairly standard assumption in the study of tensor decomposition algorithms.
For a more thorough discussion of this issue, see \cite[Section~1.4]{koi24overcomplete}.

With the above algorithmic result at hand, we can state our third result: an efficient algorithm which finds all the matrices of minimum rank in subspaces of matrices with a basis satisfying certain special properties.

\begin{restatable}[Minrank algorithm]{theorem}{thmMinRank}\label{th:minrank}
Suppose that $V \subset M_{m,n}(\mathbb{K})$ is a subspace spanned by a basis $M_1,\ldots,M_p$ (the \say{hidden basis}) where the linear spaces $\Ima(M_1),\ldots,\Ima(M_p)$ are in direct sum, and where  the linear spaces $\Ima(M_1^T),\ldots,\Ima(M_p^T)$ are also in direct sum.

There is a randomized, polynomial-time algorithm (Algorithm~\ref{algo:minrank}) such that, when given as input any basis $Z_1,\ldots,Z_p$ of $V$, it correctly finds the hidden basis and it outputs %def of \rho was missing
$\rho=\min_{M \in V, M\neq 0} \rk M$, as well as matrices $A_1,\dots,A_s$ which are, up to scalar multiplication, the only matrices of rank $\rho$ in $V$.
Moreover, the matrices $A_i$ are a subset of the hidden basis.
\end{restatable}

The above theorem follows as a corollary of our decomposition algorithm. 
Given as input a basis $Z_1,\ldots,Z_p$ of $V$, we obtain a matrix-vector decomposition of the tensor $T \in \mathbb{K}^{m \times n \times p}$ obtained by ``stacking up'' the matrices $Z_1,\ldots,Z_p$ (i.e., these matrices are the 3-slices of $T$).
From this matrix-vector decomposition, we show that the direct sum conditions imply that the matrices of minimum rank in this decomposition are the desired minimum rank matrices.
A proof of the above theorem and of the aforementioned claims can be found in \cref{sec:minrank}.

\paragraph{Genericity of hypotheses.} It is worth noting that the hypotheses in all of the theorems stated above are \emph{generic properties}, in the algebraic geometric sense (which we discuss in more detail in \cref{section: preliminaries}). 
In particular, they will work for randomly chosen tensors of rank $r$ or vector spaces of matrices spanned by matrices $M_1, \dots, M_p$ randomly chosen such that
$$\sum_{i=1}^p \rk{M_i} \leq \min(m,n).$$

%========================================================
\subsection{Previous \& related works} \label{sec:previous}
%========================================================

\input{previous-work}

%% file: previous-work.tex
\paragraph{Comparison with Jennrich's algorithm}
Let $T_1,\ldots,T_p$ be the 3-slices of the input tensor $T \in \kk^{m \times n \times p}$. 
Each slice is an $m \times n$ matrix.
Assume first for simplicity that one of the slices is invertible, for instance $T_1$. 
This implies that $m=n=r$ in \cref{th:jennrich}. 
In this case, the matrices $T_k T_1^{-1}$ for $k=2,\ldots,p$ turn out to be the simultaneously diagonalizable and the $u_i$ are the eigenvectors. 
This can be exploited algorithmically as follows: we compute two random linear combinations of the  3-slices. 
With high probability all the eigenvalues of $T_a T_b^{-1}$ are distinct, and the $u_i$ are the (unique) eigenvectors. 
The $v_i$ can computed with a similar procedure, and each $u_i$ can be paired with the corresponding $v_i$ by comparing the respective eigenvalues. Finally, the $w_i$ can be obtained by solving a linear system.
When $\kk = \rr$ or $\kk = \cc$ and we no longer assume that $m=n=r$, $T_b^{-1}$ can be replaced by the Moore-Penrose pseudoinverse $T_b^{\dagger}$.
For more details on the resulting \say{Jennrich} or \say{simultaneous diagonalization} algorithm refer to \cref{sec:simdiag}, where we have followed the presentation in~\cite[Section 3.3]{moitra18}.

As explained above, for our decomposition algorithm we drop condition~(iii) in \cref{th:jennrich}. 
The resulting algorithm can be thought of as a version of Jennrich's for multiple eigenvalues (the relevant matrices can have multiplicity greater than 1).
Also, we no longer have access to the Moore-Penrose inverse since we aim for a uniqueness theorem and a corresponding decomposition algorithm that apply to arbitrary fields. 
For this reason, instead of diagonalizing  $T_a T_b^{\dagger}$ we solve the generalized eigenvalue problem $T_a x = \lambda T_b x$. 
For the field of real and complex numbers, we also provide a version of the algorithm which uses the Moore-Penrose inverse (compare \cref{sec:images2,sec:images}).
Even when there are no multiple eigenvalues, as explained in \cref{rem:diff} this version of the algorithm slightly differs from the standard Jennrich algorithm as presented in~\cite[Section 3.3]{moitra18}.

When the matrices $T_k T_1^{-1}$ are simultaneously diagonalizable, they must in particular commute. 
In this paper we generalize Jennrich's uniqueness theorem and the corresponding decomposition algorithm from the case where there are no multiple eigenvalues to the case where the multiplicities can be larger than 1. 
In a forthcoming paper we plan to generalize this one step further, from the diagonalizable case to the case where the relevant matrices may no longer be diagonalizable, but still commute.

Another natural question to pursue is the computation of matrix-vector decompositions in the overcomplete setting. Two algorithms for the decomposition of generic tensors in that setting were recently given in~\cite{koi24overcomplete,kothari24}. In particular,
the algorithm in~\cite{koi24overcomplete} uses Jennrich's algorithm to decompose an auxiliary tensor $T'$ computed from the input tensor $T$.
If we can instead decompose $T'$ with the algorithm from the present paper, this would likely increase the range of applicability of the algorithm in~\cite{koi24overcomplete} (and of the corresponding uniqueness theorem).

\paragraph{Comparison with \cite{johnston23}.} 
The work of \cite{johnston23} studies the computational problem of determining the intersection of an algebraic variety with a generic linear subspace of appropriate dimension, such that the intersection is zero-dimensional.
Their motivation to study this problem is due to the fact that special cases of this problem have applications in quantum information theory and tensor decompositions.
Two of the main applications of the technical results of \cite{johnston23}, which are also related to our work, are \cite[Corollary 4]{johnston23} on the min-rank problem and \cite[Corollary 8]{johnston23} on the tensor decomposition problem. 

In Corollary 4, they give an algorithm that
%, when given as input 
takes as inputs $r \in \nn$, dimensions $n_1, n_2 > r$, and a basis for a linear subspace $\mathcal{U} \subset \kk^{n_1 \times n_2}$ of dimension 
$$R \leq \dfrac{\binom{n_1}{r+1} \cdot \binom{n_2}{r+1}}{(r+1)! \cdot \binom{n_1n_2 + r}{r+1}} \cdot (n_1n_2 + r) \leq \dfrac{n_1n_2}{((r+1)!)^2} $$
% \RO{see written notes for this}
where $\mathcal{U}$ satisfies $\mathcal{U} := \Kspan{A_1, \dots, A_s} + \Kspan{B_1, \dots, B_{R-s}}$. %where 
Here, the $A_i$ are \emph{generic} matrices of rank at most $r$ and the $B_i$ are \emph{generic} matrices in $\kk^{n_1 \times n_2}$ (both $A_i, B_j$ are not known in advance).
The authors prove that $\mathcal{U}$ intersects the space of matrices of rank at most $r$ at exactly $s$ matrices (up to scalar multiples) and \cite[Algorithm 1]{johnston23} returns these elements (up to scalar multiples) in $(n_1 n_2)^{O(r)}$ time.

While their result yields polynomial time algorithms for any constant $r$, and in this regime for spaces of dimension $O(n_1n_2)$, as soon as $r$ is not constant the running time of their algorithm rapidly deteriorates, and as soon as $r \sim \frac{1}{2} \cdot \log(n_1n_2)$ their algorithm stops working due to the rank condition never being satisfied.
In comparison, our \cref{th:minrank} (and \cref{rem:minrank}) works for generic instances with input $\mathcal{U} := \Kspan{A_1, \dots, A_R}$ where the matrices $A_i$ (which can be thought of as the \say{hidden basis}) satisfy the following inequality: 
$$ \sum_{i=1}^R \rk A_i \leq \min\{n_1, n_2\}. $$
In the above case, our algorithm fully recovers the hidden basis in polynomial time.
In particular, our algorithm improves \cite[Corollary 4]{johnston23} in the \say{pure low rank} version of their problem (i.e., when $\mathcal{U}$ is only generated by low-rank matrices -- the $A_i$'s) whenever $r \gtrsim \log(n_1n_2)$. 

In \cite[Corollary 8]{johnston23}, the authors show that their main algorithm (\cite[Algorithm 1]{johnston23}) can compute in randomized polynomial time the tensor rank along with the tensor rank decomposition of generically chosen tensors $T \in \kk^{n_1} \otimes \kk^{n_2} \otimes \kk^{n_3}$ of rank upper bounded by $\min\left\{ \dfrac{(n_1 - 1)(n_2-1)}{4}, n_3 \right\}$. 
Their main algorithm is remarkably simple in nature, as it simply consists of applying Jennrich's algorithm to a suitable 3-tensor (in spaces of larger dimensions) constructed from $T$.
The generic conditions in their paper (\cite[Proposition 25]{johnston23}) are distinct from the ones we consider in this paper, and it is easy to see that no generic condition is a subset of the other whenever the parameters allow for both to occur. 
However, it is worth noting that their genericity assumption holds for random tensors of rank upper bounded by $\min\left\{ \dfrac{(n_1 - 1)(n_2-1)}{4}, n_3 \right\}$, which makes their elegant algorithm useful beyond what our approach can handle.

%% file: preliminaries.tex
In this section we establish the technical preliminaries that we will need in the rest of the paper.
We begin with a discussion and precise definition of generic properties.

\paragraph{Genericity.} 
Throughout the paper, we use the term ``generic'' in its standard algebro-geometric sense: a set is generic if it contains a Zariski-open set (i.e., if it contains the complement of an algebraic set). 
More precisely, if $\kk$ is an infinite field,\footnote{If $\kk$ is finite, we can replace it by its algebraic closure.} we say that a property generically holds true in $\kk^N$ if there is a nonidentically zero polynomial $P(x_1,\ldots,x_N)$ such that the property holds true for all $x \in \kk^N$ such that $P(x_1,\ldots,x_N) \neq 0$.
In particular, if $\kk$ is the field of real or complex numbers, the set of points that do not satisfy the property has measure 0. 
 
We also use \say{generic linear combinations} in the presentation of the decomposition algorithm in order to emphasize the connection with the uniqueness theorem. 
It is more algorithmically realistic to choose the coefficients $a_1,\ldots,a_k$ from a finite set. 
This yields a randomized algorithm, and its probability of error can be made as small as desired by increasing the size of this set (Remark~\ref{rem:random}).

%==================================================
\subsection{The Generalized Eigenvalue Problem}
\label{sec:eigen}
%==================================================

\begin{definition}
For two matrices $A,B \in M_{m,n}(\kk)$, a nonzero vector $v \in \kk^n$ is a (generalized) eigenvector of the pair $(A,B)$ if there exists $\lambda \in \kk$ such that $Av = \lambda Bv$. Here, $\lambda$ is the corresponding (generalized) eigenvalue.
\end{definition}
The set of matrices of the form $A-\lambda B$ where $\lambda$ ranges over $\kk$ is traditionally called a (linear) {\em matrix pencil}. We use the same notation $(A,B)$ to denote a pair of matrices or the corresponding matrix pencil.

The (generalized) eigenspace associated to the (generalized) eigenvalue $\lambda$ is the set of vectors  $v \in \kk^n$ 
such that $Av = \lambda Bv$. This is indeed a linear subspace, like a traditional  eigenspace.
Despite this similarity, generalized eigenvectors and eigenvalues behave somewhat differently than the traditional ones. For instance, a nonzero vector
in $\ker A \cap \ker B$ can be viewed as an eigenvector associated to the eigenvalue $\lambda$ for {\em any} $\lambda \in \kk$.
These vectors will therefore belong to the eigenspace associated to  $\lambda$ for all $\lambda \in \kk$.
This motivates the following definition.
\begin{definition}
We say that an eigenvalue $\lambda$ of the pair $(A,B)$ is {\em nontrivial} if there exists an eigenvector $v$ associated to 
$\lambda$ such that $Bv \neq 0$.
\end{definition}
In particular, if $\ker B = \{0\}$ all eigenvalues are nontrivial. If $\ker B \neq \{0\}$, the pencil $(A,B)$
 is said to have {\em eigenvalues at infinity}. If $\ker A \cap \ker B \neq \{0\}$ the pencil is said to be degenerate.
 The pencil $(A,B)$ can also be viewed as a tensor of format $m \times n \times 2$. As pointed out in~\cite[Exercise~19.9]{BCS}, a pencil is nondegenerate if and only if the 2-slices of the corresponding tensor are linearly independent (the corresponding tensor is 
 said to be {\em 2-concise} in this case). We will not use this fact in the sequel.

When $m=n$ and $B$ is invertible, the eigenvalues/eigenvectors of the pair $(A,B)$ are those of the matrix $B^{-1}A$.
Applying this observation requires the computation of a matrix inverse.
An inverse-free algorithm for matrix pencil diagonalization can be found in~\cite{demmel23}.
In full generality, the generalized eigenvalue problem can be solved
by computing the Kronecker normal form of the matrix pencil (see e.g.~\cite{vanDooren79}).
In our case, however, there is a much simpler solution due to the
special structure of the matrices involved (see Section~\ref{sec:images}
and especially Proposition~\ref{prop:eigen} in Section~\ref{sec:images2}).

%==================================================
\subsection{Facts about matrix-vector decompositions}\label{sec:matrix}
%==================================================

We now establish basic facts about matrix-vector decompositions %add:
(recall that this notion was introduced in Definition~\ref{def:matrix_vector}).

\begin{proposition} \label{prop:smallest}
The smallest rank of a matrix-vector decomposition is equal to the tensor rank of $T$.
\end{proposition}
\begin{proof}
Suppose that $T$ has a matrix-vector decomposition of rank $r$. Writing down each matrix $M_l$ in this decomposition 
as a sum of $\rk(M_l)$ matrices of rank one shows that $\rk(T) \leq r$. Assume conversely that $T$ admits a decomposition as a sum of $r$ rank-1 tensors  as in~(\ref{eq:decomp}). If there are no colinear vectors in the list $w_1,\ldots,w_l$, (\ref{eq:decomp}) is already a matrix-vector decomposition with $M_i=u_i \otimes v_i = u_i v_i^T$.
 If some of these vectors are colinear, 
 we can factorize using the rule:
$$\sum_{\ell} u_{\ell} \otimes v_{\ell} \otimes (\lambda_{\ell} w) =
 (\sum_{\ell} \lambda_{\ell}  u_{\ell} \otimes v_{\ell}) \otimes w.$$
 This results in a matrix-vector decomposition of rank at most $r$.
\end{proof}

Let $r=\rk(T)$. We say that $T$ has a unique matrix-vector decomposition of rank $r$ if up to permutation, 
the terms $M_{\ell} \otimes w_{\ell}$ are the same in all  matrix-vector decompositions of $T$ of rank $r$.
We will show in Proposition~\ref{prop:rk12matrix} that uniqueness of decomposition as a sum of rank-1 tensors implies
uniqueness of matrix-vector decompositions. The proof relies on the following well-known lemma (see e.g.~\cite{leurgans93}).
\begin{proposition} \label{lem:rank2}
The decomposition of a matrix of rank 2 as the sum of two matrices of rank 1 is never unique.
\end{proposition}
For completeness, we give a proof of this lemma below (note that it applies to an arbitrary field).
\begin{proof}
Let $M=u_1 v_1^T+u_2 v_2^T$ be a matrix of rank 2. Since $M=u_1(v_1+v_2)^T+(u_2-u_1)v_2^T$, it suffices to show that
$$\{u_1 v_1^T,u_2 v_2^T\} \neq \{u_1(v_1+v_2)^T,(u_2-u_1)v_2^T\}.$$
Suppose that $u_1 v_1^T = u_1(v_1+v_2)^T$. This implies $u_1v_2^T=0$, i.e., $u_1=0$ or $v_2=0$. This is impossible since $M$ would then be of rank less than 2. The equality $u_1 v_1^T = (u_2-u_1)v_2^T$ is also impossible. Indeed, the columns of
$u_1 v_1^T$ are colinear to $u_1$ and those of $(u_2-u_1)v_2^T$ are colinear to $u_2-u_1$. Then $u_1$ would be colinear 
to $u_2$, and $M$ would again be of rank less than 2. 
\end{proof}

\begin{proposition} \label{prop:rk12matrix}
Suppose that a tensor of rank $r$ admits a unique decomposition 
\begin{equation} \label{eq:unique}
T=\sum_{\ell=1}^r (u_{\ell} \otimes v_{\ell}) \otimes w_{\ell}
\end{equation}  as a sum of $r$ tensors of rank one.
This is also the unique matrix-vector decomposition of $T$ of rank $r$. In particular, there are no colinear vectors in
the list $w_1,\ldots,w_{\ell}$.
\end{proposition}
\begin{proof}
We first show that there are no colinear vectors in this list. 
 Assume for instance that $w_1$ and $w_2$ are colinear. We can assume that $w_1=w_2$ by scaling 
$v_1$ and $w_1$ (or $v_2$  and $w_2$)
 if necessary. The sum of the first two terms in~(\ref{eq:unique}) is equal to $M \otimes w_1$ 
where $M=u_1 \otimes v_1 + u_2 \otimes v_2$. This matrix must be of rank 2 by the minimality of~(\ref{eq:unique}).
By Lemma~\ref{lem:rank2}, there is another decomposition $M=u'_1 \otimes v'_1 + u'_2 \otimes v'_2$.
Replacing the first two terms of~(\ref{eq:unique}) by $u'_1 \otimes v'_1 \otimes w_1$ and $u'_2 \otimes v'_2 \otimes w_1$
yields a rank $r$ decomposition of $T$ which is not equivalent to~(\ref{eq:unique}).

We have thus shown that there are no colinear vectors among $w_1,\ldots,w_{\ell}$; hence~(\ref{eq:unique}) is a bona fide matrix-vector decomposition. It remains to show that any other matrix-vector decomposition of rank $r$, say,
\begin{equation} \label{eq:unique2}
T=\sum_{\ell} M_{\ell} \otimes w'_{\ell}
\end{equation}
must be equivalent to~(\ref{eq:unique}). We can expand each $M_{\ell}$ as a sum of $\rk(M_{\ell})$ matrices of rank 1. From these expansions and~(\ref{eq:unique2}) we obtain a decomposition of $T$ as a sum of $r$ tensors of rank 1.
This new decomposition must be equivalent to our first decomposition~(\ref{eq:unique}), which is assumed to be unique.
But: 
\begin{itemize}
\item[(i)] there are no colinear $w_{\ell}$ in~(\ref{eq:unique});
\item[(ii)] whereas the expansion of $M_{\ell}$ yields $\rk(M_{\ell})$ rank 1 tensors with the same third-mode vector $w'_{\ell}$ in our new decomposition of $T$ as a sum of rank 1 tensors.
\end{itemize}
Hence $\rk(M_{\ell}) = 1$ for all $\ell$. It follows that~(\ref{eq:unique}) and~(\ref{eq:unique2}) are equivalent as decompositions as sums of rank 1 tensors, and also as matrix-vector decompositions.
\end{proof}

%% file: decomposition-algorithm.tex
In this section we propose and analyze an algorithm which, given an input tensor $T$, computes the unique matrix-vector decomposition guaranteed by \cref{th:unique1,th:unique2}. 
We will take the point of view of \cref{th:unique2}: assuming that there is a matrix-vector decomposition
\begin{equation} \label{eq:uniqueth2}
T=\sum_{\ell=1}^q M_{\ell} \otimes w_{\ell}
\end{equation}
where  the linear spaces $\Ima(M_{\ell})$ are in direct sum  { and where the linear spaces $\Ima(M_{\ell}^T)$ are also in direct sum}, we want to compute that (unique) decomposition.
In \cref{sec:images2}, we first show how to compute the spaces $\Ima M_\ell$. 
In \cref{sec:disjoint} we describe an algorithm for the case where the subspaces $\Ima(M_{\ell})$ have a very simple form, i.e., when they are coordinate subspaces.
Finally, in \cref{sec:main} we combine the algorithms from \cref{sec:images2,sec:disjoint} to give our main algorithm.

%========================================================
\subsection{Computing  \texorpdfstring{$\Ima(M_{\ell})$}{ImaMell} in an arbitrary field} \label{sec:images2}
%========================================================

In this section $\mathbb{K}$ can be an arbitrary (infinite) field.\footnote{In \cref{rem:random} we explain how the algorithm can be adapted to finite fields.}  
As explained in \cref{sec:previous}, we only need to assume that we have access to an algorithm for the computation of roots of polynomials with coefficients in $\mathbb{K}$. 
For the field of real and complex numbers, 
we provide in \cref{sec:images} an alternative algorithm based on the Moore-Penrose inverse, in the same style as the classical simultaneous diagonalization algorithm (see \cref{sec:simdiag}).

We will compute the linear spaces $\Ima(M_1),\ldots,\Ima(M_q)$ by solving a generalized eigenvalue problem\footnote{Recall that this was the point of view of \cref{sec:eigen,sec:proof}.} using the following result.

\begin{proposition} \label{prop:eigen}
Let $A=UA'V^T, B=UB'V^T$ where $U \in M_{m,r}(\mathbb{K})$, $V \in M_{n,r}(\mathbb{K})$ and $A',B' \in M_{r,r}(\mathbb{K})$. 
Assume that $r \leq \min(m,n)$. Then:
\begin{itemize}
\item[(i)] All the minors of size $r$ of $A-\lambda B$ are scalar multiples of $\det(A'-\lambda B')$.
\item[(ii)] Assume moreover that $U,V,A',B'$ are all of rank $r$.
Let $M_a \in M_{r,r}(\mathbb{K})$ be any submatrix of $A$ of rank $r$, 
and let $M_b$ be the matching submatrix of $B$ (i..e,  we select the same  rows and columns as for $M_a$). 
The nontrivial eigenvalues of the pair $(A,B)$ are exactly the roots 
of the polynomial $P(\lambda)=\det(M_a-\lambda M_b)$.
\end{itemize}
\end{proposition}
\begin{proof}
A submatrix of size $r$ of  $A-\lambda B$ is of the form $M_a-\lambda M_b$ where $M_a \in M_{r,r}(\mathbb{K})$ is a submatrix of $A$
and $M_b$ is the matching submatrix of $B$.
Note that 
\begin{equation} \label{eq:submatrix}
M_a = U'A'V',\ M_b = U'B'V'
\end{equation}
where $U', V'$ are $r \times r$ submatrices of $U$ and $V^T$. This implies:
\begin{equation} \label{eq:detprod}
P(\lambda) =
\det(M_a-\lambda M_b) =  (\det U') (\det V') \det (A'-\lambda B'),
\end{equation}
and we have proved the first part of the proposition.

For the proof of (ii), we first observe that $\rk A = r$ since $U,A',V$ are all of rank $r$ (and likewise, $\rk B =r$). 
Hence there exists a submatrix 
$M_a$ of rank~$r$. Also, we observe that $\ker A = \ker B = \ker V^T$ and this kernel is of dimension $n-r$ since $U,V,A',B'$ are all of full rank $r$. The polynomial $P(\lambda)=\det(M_a-\lambda M_b)$ is not identically 0 since $P(0) \neq 0$.

Let $\lambda$ be a nontrivial eigenvalue of $(A,B)$ and $V_{\lambda}$ be the corresponding eigenspace.
By definition, there is $v {\not \in} \ker B$ such that $(A - \lambda B)v = 0$. 
Note that $\ker B \subseteq V_{\lambda}$ since $\ker A= \ker B$. 
Therefore, $V_{\lambda}$ must be of dimension at least $n-r+1$, i.e., $A-\lambda B$ must be of rank at most $r-1$. 
Its submatrix $M_a - \lambda M_b$ must therefore be of rank at most $r-1$, hence $P(\lambda)=0$.

Assume conversely that $P(\lambda)=0$. 
By \cref{eq:submatrix}, $U'$ and $V'$ are of rank $r$ since $M_a$ is of rank $r$.
Hence $\det (A'-\lambda B')=0$ by \cref{eq:detprod}. 
By part (i) of the proposition, this implies $\rk(A-\lambda B) \leq r-1$.
This in turn implies that $\lambda$ is a nontrivial eigenvalue of $(A,B)$ by the converse of the argument in the preceding paragraph. 
Namely, since $\ker (A-\lambda B)$ is of dimension at least $n-r+1$ and $\ker B$ of dimension $n-r$, there must exist $v$ such that $(A-\lambda B)v=0$ and $Bv \neq 0$.
\end{proof}

We are now ready to state the algorithm to compute the spaces $\Ima(M_\ell)$.

\begin{algorithm}\label{algo:images}
\SetAlgoLined
\nonl \textbf{Input:} a tensor $T \in \mathbb{K}^{m \times n \times p}$  with an unknown matrix-vector decomposition  
$T=\sum_{\ell =1}^q M_{\ell} \otimes w_{\ell}$.\\
\nonl \textbf{Output:} The linear spaces $\Ima(M_1),\ldots,\Ima(M_q)$.
% (this decomposition is guaranteed to be correct under the hypotheses of Theorem~\ref{th:0decomp}.\\

Compute two generic linear combinations $T_a,T_b \in M_{m,n}(\mathbb{K})$ of the $3$-slices of $T$.\\
Compute $r=\rk(T_a)$ and find a submatrix $M_a$ of $T_a$ of size $r$ and rank $r$.\\
Let $M_b$ be the matching submatrix of $T_b$, and $P(\lambda)=\det(M_a-\lambda M_b)$.\\
Compute the roots $\lambda_1,\ldots,\lambda_q$ of $P$.\\
Compute the corresponding eigenspaces $V_1=\ker(T_a-\lambda_1 T_b),\ldots,V_q=\ker(T_a-\lambda_q T_b)$.\\
Output $T_b(V_1),\ldots,T_b(V_q)$.

\caption{Computing the linear spaces $\Ima(M_1),\ldots,\Ima(M_q)$.}
\end{algorithm}

The next theorem proves the correctness of the above algorithm.

\begin{theorem} \label{th:images}
Suppose that a tensor $T \in \mathbb{K}^{m \times n \times p}$ has a matrix-vector decomposition of the form:
$$T=\sum_{\ell=1}^q M_{\ell} \otimes w_{\ell}$$
where the linear spaces $\Ima(M_{\ell})$ are in direct sum (and where the linear spaces $\Ima(M_{\ell}^T)$ are also in direct sum.) 

Then Algorithm~\ref{algo:images} on input $T$ and with vectors $a,b \in \mathbb{K}^p$ of the linear combinations computed at step 1 being generically chosen, Algorithm~\ref{algo:images} correctly outputs the linear spaces $\Ima(M_1),\ldots,\Ima(M_q)$.
\end{theorem}
\begin{proof}
By \cref{prop:eigen}, at step 4 we compute the nontrivial eigenvalues of $(T_a,T_b)$. In this application of 
\cref{prop:eigen}, the matrices $A',B'$ are diagonal and of full rank since $a$ and $b$ are generically chosen.
There are exactly $q$ nontrivial eigenvalues by \cref{lem:image2}.(ii).	
We compute the corresponding eigenspaces at step~5, and another application of \cref{lem:image2}.(ii) shows that the output of step 6 is correct.
\end{proof}

\begin{remark} \label{rem:random}
In Algorithm~\ref{algo:images} the coefficients of the linear combination can be drawn uniformly at random from a finite set $S$. 
The proof of \cref{th:images} reveals that these coefficients should avoid the zero sets of polynomially many polynomials of polynomially bounded degree. 
By the Schwartz-Zippel Lemma~\cite{Schw,zippel}, we can make the probability of error smaller
than, say, 1/3 (or any other constant) by taking $S$ of polynomial size. 
This remark also applies to finite fields: we can take $S \subseteq \mathbb{K}$ if $\mathbb{K}$ is large enough.
If not, we can take the elements of $S$ from a field extension.
\end{remark}

A similar remark applies to the computation of $\Ima(M_1),\ldots,\Ima(M_q)$ using the Moore-Penrose inverse (\cref{prop:image}).

%========================================================
%\subsection{The block diagonal case}
\subsection{Disjoint rows}\label{sec:disjoint}
%========================================================

The assumption that the spaces  $\Ima(M_{\ell})$ are in direct sum clearly holds in the  special case where  
the only nonzero rows of $M_1$ are its first $\rk(M_1)$ rows, the only nonzero rows of $M_2$ are the next $\rk(M_2)$ rows (namely, rows $1+\rk M_1$ to $\rk M_2 + \rk M_1$), and so on.
%The assumption that the spaces  $\Ima(M_{\ell}^T)$ are in direct sum is likewise satisfied when the only nonzero columns of $M_1$ are its first $\rk(M_1)$ columns and the only nonzero columns of $M_2$ are the next $\rk(M_2)$ columns, etc\ldots
%If this special structure holds for the rows and columns of the $M_{\ell}$ we say that $T$ satisfies the {\em block diagonal assumption}.
In this case we say that $T$ has the {\em disjoint rows property}. From \cref{eq:slicematrix}, the slices $Z_1,\ldots,Z_p$ must then have the structure:
\begin{equation} \label{eq:disjoint}
Z_k =  \begin{pmatrix}
Z_{k1} \\
\vdots \\
Z_{kq}\\
0
\end{pmatrix} =  \begin{pmatrix}
w_{1k} M'_1 \\
\vdots \\
w_{qk} M'_q\\
0
\end{pmatrix}
\end{equation}
where $M'_{\ell}$ is the $\rk(M_{\ell}) \times n$ block where all the nonzero entries of $M_{\ell}$ are located. The block of zeros at the bottom of $Z_k$ is present only when $m > \rk(T)=\sum_{\ell=1}^q \rk(M_{\ell})$;  $Z_{k\ell}$ is 
the block of $\rk(M_{\ell})$ rows of $Z_k$ which matches the corrresponding block on the right-hand side of \cref{eq:disjoint}.
Note that  $Z_{k\ell} = w_{\ell k} M'_{\ell}$, i.e., the blocks of $Z_1,\ldots,Z_p$ that are in same position in \cref{eq:disjoint} are all  proportional, and the $w_{\ell k}$ are the coefficients of proportionality. 
This leads to a very simple decomposition algorithm: 

\begin{algorithm} \label{algo:disjointrows}
\SetAlgoLined
\nonl \textbf{Input:} a tensor $T \in \mathbb{K}^{m \times n \times p}$  with the disjoint rows property, 
and the values $\rk(M_1),\ldots,\rk(M_q)$.\\ 
\nonl \textbf{Output:} The matrix-vector decomposition $T=\sum_{\ell =1}^q M_{\ell} \otimes w_{\ell}$.
% (this decomposition is guaranteed to be correct under the hypotheses of Theorem~\ref{th:0decomp}.\\

For $\ell=1$ to $q$:\\
\Indp  Find a slice $Z_k$ of $T$ such that $Z_{k \ell} \neq 0$.\\
 Set $w_{\ell k}=1$, $M'_{\ell}= Z_{k \ell}$.\\
 For all $j \neq k$ set $w_{\ell j}$ so that $Z_{j\ell} = w_{\ell j} M'_{\ell}$.\\
 Construct $M_{\ell}$ by putting the appropriate number of null rows above and below $M'_{\ell}$.\\
%\caption{Jennrich's algorithm for tensors of rank $n$.}
\Indm Output the decomposition $T=\sum_{\ell =1}^q M_{\ell} \otimes w_{\ell}$.
\caption{Decomposition of a tensor with disjoint rows.}
\end{algorithm}

\begin{proposition}
If $T$ has the disjoint rows property, Algorithm~\ref{algo:disjointrows} produces a correct matrix-vector decomposition
 $T=\sum_{\ell =1}^q M_{\ell} \otimes w_{\ell}$.
\end{proposition}
\begin{proof}
At line 2 of the algorithm, we look for a slice $Z_k$ where $Z_{k \ell} \neq 0$. 
There must be such a slice since it is assumed that $T$ has a matrix decomposition $T=\sum_{\ell =1}^q M_{\ell} \otimes w_{\ell}$ with the disjoint rows property. 
Indeed, $Z_{k \ell} = 0$ for all $k$ implies $w_{\ell}=0$ or $M_{\ell} = 0$. 
This is not allowed by definition of a matrix-vector decomposition. At line 3 we set $w_{\ell k}=1$. This is legitimate since $w_{\ell}$ and $M_{\ell}$ are only unique up to scaling. 
The correctness of the algorithm then follows from \cref{eq:disjoint}.
 \end{proof}

\begin{remark}
In Algorithm~\ref{algo:disjointrows}, we assumed that $\rk(M_1),\ldots,\rk(M_q)$ are given as input to the algorithm.
We can do this, since we have computed these values (and more) in \cref{sec:images2}. 
\end{remark}

%======================================================================
\subsection{Main algorithm}\label{sec:main}
%======================================================================

Our main algorithm will reduce the general case to the case of disjoint rows treated in \cref{sec:disjoint}. 
This is easy  once we have determined $\Ima(M_1),\ldots,\Ima(M_q)$. 
Indeed, we can apply a linear map which sends each $\Ima(M_{\ell})$ to $E_{\ell}$ where $E_1$ is the space spanned by the first $\rk(M_1)$ vectors of the canonical basis of $\mathbb{K}^m$, $E_2$ is spanned by the next $\rk(M_2)$ vectors of this basis, etc.
This is justified by \cref{lem:mult}. 
We are now in position to describe our main algorithm and prove our main algorithmic result, and we restate the latter here for convenience.

\begin{algorithm}\label{algo:main}
\SetAlgoLined
\nonl \textbf{Input:} a tensor $T \in \mathbb{K}^{m \times n \times p}$  with an unknown matrix-vector decomposition  
$T=\sum_{\ell =1}^q M_{\ell} \otimes w_{\ell}$.\\
\nonl \textbf{Output:} The above decomposition.

Determine the linear spaces $\Ima(M_1),\ldots,\Ima(M_q)$ using Algorithm~\ref{algo:images} (or via \cref{prop:image}).

Find $A \in GL_m(\mathbb{K})$ such that $A$ maps $\Ima(M_1),\ldots,\Ima(M_q)$ to the linear spaces $E_1,\ldots,E_q$ defined at the beginning of \cref{sec:main}.

Let $T'$ be the tensor with slices $AT_1,\ldots,AT_p$, where $T_1,\ldots,T_p$ are the slices of $T$.

Compute a matrix-vector decomposition $T'= \sum_{\ell=1}^q N_{\ell} \otimes w_{\ell}$ with Algorithm~\ref{algo:disjointrows}.

Output the decomposition $T=\sum_{\ell = 1}^q (A^{-1}N_{\ell}) \otimes w_{\ell}$.

\caption{Matrix-vector decomposition algorithm.}
\end{algorithm}

\uniquenessAlg*

\begin{proof}
At Step 2 we find the required matrix $A$ since the spaces $\Ima(M_{\ell})$ are in direct sum.
The tensor $T'$ defined at step 3 has the disjoint rows property by \cref{lem:mult}.
Thereby, we can decompose it with Algorithm~\ref{algo:disjointrows} at Step 4.
Finally, at step 5 we undo the effect of the multiplications by $A$.
\end{proof}

\begin{remark} \label{rem:diff}
In Algorithm~\ref{algo:main} we could even reduce to a block diagonal structure (i.e., disjoint rows and distinct columns) 
by multiplying the slices of $T$ from the left {\em and from the right}.
But as demonstrated in Section~\ref{sec:disjoint} the disjoint rows property alone is sufficient, so we only multiply from 
the left.
\end{remark}

\begin{remark}
This algorithm departs from  the standard version of Jennrich's algorithm (as presented in~\cite[Section 3.3]{moitra18} and \cref{sec:simdiag}) even in the case where the $M_{\ell}$ have rank 1 and when we use the Moore-Penrose inverse (\cref{prop:eigen}) to determine the spaces $\Ima(M_{\ell})$ at step~1. 
Indeed, we only compute one pseudo-inverse whereas the standard algorithm computes two pseudo-inverses.
Moreover, we do not need to solve an overdetermined system of linear equations to find the vectors $w_i$ (as in step 5 of Algorithm~\ref{algo:jennrich}). 
Instead, their components are directly read off as coefficients of proportionality in Algorithm~\ref{algo:disjointrows}.
\end{remark}

\begin{remark}\label{rem:time}
Suppose that $\mathbb{K}=\mathbb{Q}$ and that the coefficients of the linear combinations $T_a,T_b$ in Algorithm~\ref{algo:images} are chosen at random from a polynomial size set, as suggested in \cref{rem:random}.
Then Algorithm~\ref{algo:main} can be implemented efficiently in the Turing machine model of computation: we obtain a randomized algorithm that runs in time polynomial in the bit size of the input tensor. Indeed, at step 4 of Algorithm~\ref{algo:images} we compute the nontrivial generalized eigenvalues of the pair $(T_a,T_b)$ as the roots of $P(\lambda)$.
By \cref{lem:image2}.(ii), these eigenvalues are of the form  
$\lambda=\langle a,\gamma_{q} \rangle /  \langle b,\gamma_{q} \rangle$, i.e., they are rational numbers.
But it is well known that rational roots of polynomials with rational coefficients can be computed in polynomial time, so 
we can compute the eigenvalues in polynomial time. The other steps of Algorithms~\ref{algo:disjointrows},~\ref{algo:images} and~\ref{algo:main} are standard linear algebraic computations that run in polynomial time.
\end{remark}

\begin{remark}\label{remark: uniqueness algorithm via adjoint algebra}
Similarly to \cref{remark: adjoint algebra uniqueness}, we note here that one can also compute the minimum matrix-vector decomposition by using the indecomposable vector space decomposition algorithm as a subroutine.
By block-diagonalizing the adjoint algebra and computing some generalized eigenvalues, we obtain a decomposition for each of the matrices $M_\ell$, and we can proceed as we did above to compute the $w_\ell$ vectors.
\end{remark}

%% file: minrank.tex
In this section we provide an application of our uniqueness theorem and algorithmic result: finding the matrices of minimum rank in certain generic vector spaces of matrices.
We begin with the following proposition on the minimum ranks of certain vector spaces.

% \begin{proposition} \label{prop:smallestrank}
% Let $V$ be a subspace of $M_{m,n}(\mathbb{K})$, and let $M_1,\ldots,M_p$ be a basis of $V$.
%  If the subspaces $\Ima M_1,\ldots,\Ima M_p$ are in direct sum, then 
%  $$\min_{M \in V, M \neq 0} \rk M = \min_{1 \leq i \leq p} \rk M_i.$$
% \end{proposition}

% \begin{proof}
% For any nonzero matrix $M \in V$ we can write $M=\sum_{j \in I} \alpha_j M_j$ where $I \subseteq \{1,\ldots,p\}$ is nonempty, and $\alpha_j \neq 0$ for all $j$.
% If $x \in  \ker M$ then $\sum_{j \in I} \alpha_j M_jx =0$; this implies $M_j x =0$ for all $j$ by the direct sum assumption.
% We have shown that $\ker M = \bigcap_{j \in I} \ker M_j$. In particular, we have $\ker M \subseteq \ker M_j$ for any $j \in I$,
% so that $\displaystyle \rk M \geq \rk M_j \geq \min_{1 \leq i \leq p} \rk M_i$ by the rank-nullity theorem.
% \end{proof}

\begin{proposition}\label{th:smallest}
Let $V$ be a subspace of $M_{m,n}(\mathbb{K})$, and let $M_1,\ldots,M_p$ be a basis of $V$.
We assume that the matrices in this basis are ordered by nondecreasing rank ($\rk M_i \leq \rk M_{i+1}$).
If the subspaces $\Ima M_1,\ldots,\Ima M_p$ are in direct sum, then 
$$\min \{\rk M;\ M \in V \setminus \text{Span}(M_1,\ldots,M_{i-1})\} = \rk M_i, \quad \text{for all } i \geq 1.$$
\end{proposition}

\begin{proof}
Let $r_i := \min \{\rk M;\ M \in V \setminus \text{Span}(M_1,\ldots,M_{i-1})\}$.
Since $M_i \in V \setminus \text{Span}(M_1,\ldots,M_{i-1})$, we have $r_i \leq \rk M_i$.
In order to prove the converse inequality, pick any matrix $M \in V \setminus \text{Span}(M_1,\ldots,M_{i-1})\}$.
We can write $M=\sum_{j \in [i-1]} \beta_j M_j + \sum_{j \in I} \alpha_j M_j$ where $I \subseteq \{i,\ldots,p\}$ is nonempty, and $\alpha_j \neq 0$ for all $j \in I$.
If $x \in  \ker M$ then $\sum_{j \in [i-1]} \beta_j M_jx + \sum_{j \in I} \alpha_j M_jx =0$; this implies $M_j x =0$ for all $j \in I$ by the direct sum assumption.
Thus, $\ker M \subseteq \bigcap_{j \in I} \ker M_j$. 
In particular, we have $\ker M \subseteq \ker M_j$ for any $j \in I$, so that $\displaystyle \rk M \geq \rk M_j \geq \min_{1 \leq i \leq p} \rk M_i$ by the rank-nullity theorem.
This completes the proof since $\rk M_j \geq \rk M_i$ for $j \geq i$.
\end{proof}

%================================================================
\subsection{Uniqueness from rank arguments}
%================================================================

With an additional assumption, we can strengthen the conclusion of \cref{th:smallest}. Namely, we can conclude that the only matrices of minimum rank in $V \setminus \text{Span}(M_1,\ldots,M_{i-1})$ are elements of the basis $M_1,\ldots,M_p$. 

\begin{proposition} \label{th:uniquerank}
Let $V$ be a subspace of $M_{m,n}(\mathbb{K})$, and let $M_1,\ldots,M_p$ be a basis of $V$.
We assume that the matrices in this basis are ordered by nondecreasing rank ($\rk M_i \leq \rk M_{i+1}$).
As in \cref{th:smallest} we assume that  the subspaces $\Ima M_1,\ldots,\Ima M_p$ are in direct sum;
in addition we assume that $\rk N^{i,j} > \rk M_i$ for all $1 \leq i < j \leq p$ where
 \begin{equation} \label{eq:Nij}
 N^{i,j}= 
\begin{pmatrix} 
M_i \\ M_j
  \end{pmatrix}
  \end{equation}
  is a $2m \times n$ matrix.
Then for all $i \geq 1$,
 \begin{equation} \label{eq:smallest}
 \min \{\rk M;\ M \in V \setminus \textrm{Span}(M_1,\ldots,M_{i-1})\} = \rk M_i,
 \end{equation}
  %(in particular, we obtain \cref{prop:smallest} for $i=1$).
  and the only matrices $M$ that achieve the minimum are %among 
  scalar multiples of $M_{i},M_{i+1},\ldots,M_p$.
\end{proposition}

\begin{proof}
We already know from \cref{th:smallest} that \cref{eq:smallest} holds true, so we only have to prove the second assertion. 
Consider therefore a matrix $M \in V \setminus \textrm{Span}(M_1,\ldots,M_{i-1})\}$ of same rank as $M_i$, and let us write $M=\sum_j \alpha_j M_j$. 
There is at least one nonzero coefficient $\alpha_{i_1}$ in this expression with $i_1 \geq i$. 
Suppose that there is at least one other nonzero coefficient $\alpha_{i_2}$ (which may be smaller or bigger than $i$). 
As we have seen in the proof of \cref{prop:smallest}, this implies $\ker M \subseteq \ker M_{i_1} \cap \ker M_{i_2}$.
Note that this intersection of kernels is defined by the linear system $N^{i_1,i_2}x=0$, and $\rk N^{i_1,i_2} > \rk M_{i_1}$ by assumption. 
This shows that $\rk M > \rk M_{i_1} \geq \rk M_i$, and $M$ would therefore not be of  same rank as $M_i$.
\end{proof}

In particular, taking $i=1$ in \cref{eq:smallest}, we see that that the nonzero matrices of smallest rank in $V$ are up to scaling the matrices in the basis $M_1,\ldots,M_p$ which have same rank as $M_1$. 
We will use this fact in \cref{sec:minrank} to analyze a minrank algorithm, and we will use it again in the appendix to given an alternative proof of the Jennrich uniqueness theorem. 
The general version of \cref{th:uniquerank} ($i \geq 1$ in \cref{eq:smallest}) is only provided for the sake of completeness, and we do not use it in the remainder of the paper.

\begin{remark}
The extra assumption in \cref{th:uniquerank} is necessary, as this new hypothesis is not always satisfied when the images are in direct sum. 
Take for instance $M_1=u_1v^T$, $M_2=u_1v^T$. 
If $u_1,u_2$ are linearly independent and $v \neq 0$, the images are in direct sum. 
But $N^{1,2}$ is of the form $uv^T$, where $u$ is obtained by stacking $v_1$ on top of $v_2$. 
So this matrix is not of higher rank than $M_1$ or $M_2$.
It is worth noting that the extra assumption will be satisfied when $\Ima(M_i)$ are in direct sum and when $\Ima(M_i^T)$ are also in direct sum.
\end{remark}

%================================================================
\subsection{The minrank algorithm} \label{sec:minrank}
%================================================================

In this section we analyze the following algorithm for the computation of all matrices of minimum rank in a given subspace of matrices (up to scalar multiples).
The algorithm simply computes a matrix-vector decomposition of the tensor formed by making the input basis its 3-slices.

\begin{algorithm}\label{algo:minrank}
\SetAlgoLined
\nonl \textbf{Input:} a subspace of $V \subseteq M_{m,n}(\mathbb{K})$, given by a basis $Z_1,\ldots,Z_p$.

\nonl \textbf{Output:} Set of matrices $A_1, \ldots, A_s \in V$ of minimum rank 

Construct the tensor $T \in  \mathbb{K}^{m \times n \times p}$ with slices $Z_1,\ldots,Z_p$.

Apply Algorithm~\ref{algo:main} to compute a matrix-vector decomposition $T=\sum_{\ell=1}^q A_{\ell} \otimes w_{\ell}$ of minimum rank, where $A_1,\ldots,A_q$ are sorted by nondecreasing rank.

Let $A_1,\ldots,A_s$ be the matrices of minimum rank among $A_1,\ldots,A_q$, and let $\rho = \rk A_1= \cdots = \rk A_s$. Declare that $$\rho=\min_{M \in V, M \neq 0} \rk M,$$  and output matrices $A_1,\dots,A_s$.

\caption{Minrank algorithm}
\end{algorithm}

\thmMinRank*

\begin{proof}
By \cref{th:smallest},  $\min_{M \in V, M \neq 0} \rk M = \min_{1 \leq i \leq p} \rk M_i.$
Moreover, by the case $i=1$ of \cref{th:uniquerank}, the only nonzero matrices in $V$ of minimum rank are up to scalar multiplication the matrices of minimum rank in the list $M_1,\ldots,M_p$.
Note that the hypothesis $\rk N^{i,j} > \rk M_i$ in \cref{th:uniquerank} is indeed satisfied 
due to the assumption that the spaces $\Ima M_{\ell}^T$ are in direct sum: we have 
$$\rk N^{i,j} = \rk (N^{i,j})^T = \rk(M_i^T)+\rk(M_j^T) > \rk M_i.$$

The correctness of the algorithm therefore follows from two claims
which we establish in the remainder of the proof:
\begin{itemize}
    \item[(i)] In the matrix-vector decomposition computed at step 2 of the algorithm,  the number of terms $q$ is equal to $p$.
    \item[(ii)] The matrices $A_1,\ldots,A_q$ are equal up to permutation and scaling to the matrices $M_1,\ldots,M_p$ of the hidden basis.
\end{itemize}

Since $Z_1,\ldots,Z_p$ and $M_1,\ldots,M_p$ are two bases of $V$, there is a change of basis matrix
 $W \in GL_p(\mathbb{K})$ such that $Z_k = \sum_{\ell =1}^p w_{\ell k} M_{\ell}$.
Hence by \cref{eq:slicematrix}, we have the decomposition 
 \begin{equation} \label{eq:basis}
 T=\sum_{\ell} M_{\ell} \otimes w_{\ell}, 
 \end{equation}
 where $w_{\ell} = (w_{\ell 1},\ldots,w_{\ell p})$. This is a matrix-vector decomposition since the $w_{\ell}$ are pairwise linearly independent (in fact, they are linearly independent since $W$ is invertible). 
By \cref{th:unique2}, this is the unique matrix-vector decomposition of $T$ of minimal rank.
The decomposition produced at step~2  of the algorithm must therefore be the same as \cref{eq:basis} up to permutation and scaling.
In particular, the two decompositions have the same number of terms (thereby proving claim (i)) and $A_1,\ldots,A_q$ are equal up to permutation and scaling to $M_1,\ldots,M_p$ (thereby proving claim (ii)).
\end{proof}

\begin{remark}\label{rem:minrank}
% Suppose that we choose linearly independent matrices $M_1,\ldots,M_p$ and give to Algorithm~\ref{algo:minrank} another basis $Z_1,\ldots,Z_p$ of $V = \text{Span}(M_1,\ldots,M_p)$.
% \cref{th:minrank} shows that for a generic choice of $M_1,\ldots,M_p$, the algorithm will determine all the matrices of minimum rank in $V$, which are up to scaling the members of the \say{hidden basis} ($M_1,\ldots,M_p)$. 
% This genericity property holds true as long as 
The decomposition $T=\sum_{\ell=1}^q A_{\ell} \otimes w_{\ell}$ computed at step 2 of the algorithm provides a certificate of correctness of its output. 
Indeed, from this decomposition we can easily check that $p=q$, that the linear spaces $\Ima(A_{\ell})$ are in direct sum, and that the linear spaces $\Ima(A_{\ell}^T)$ are also in direct sum.
These 3 conditions will be satisfied generically by \cref{th:minrank}; and whenever they are satisfied, \cref{th:minrank} guarantees that the algorithm's output is correct (note in particular that the condition $p=q$ implies that $A_1,\ldots,A_q$ is a basis of $V$).
\end{remark}

%% file: simultaneous-diagonalization.tex
In this section we recall the simultaneous diagonalization / Jennrich algorithm following~\cite[Section 3.3]{moitra18}.
It provides an efficient decomposition algorithm for generic tensors of rank $r \leq \min(m,n)$.
We assume that $K$ is the field of real or complex numbers since this version of the algorithm uses the Moore-Penrose inverse.
\begin{algorithm} \label{algo:jennrich}
\SetAlgoLined
\nonl \textbf{Input:} a tensor $T \in K^{m \times n \times p}$ satisfying the conditions of Theorem~\ref{th:jennrich}.\\ 
\nonl \textbf{Output:} the (unique) decomposition $T=\sum_{i=1}^r u_i \otimes v_i \otimes w_i$.

Compute two generic linear combinations 
$$T_a=\sum_{k=1}^p a_k T_k,\ T_b=\sum_{k=1}^p b_k T_k$$
of the slices of $T$.\\
 Compute the nonzero eigenvalues $\lambda_1,\ldots,\lambda_r$ and the corresponding eigenvectors $u_1,\ldots,u_r$ of $T_a{T_b}^{\dagger}$.\\
Compute the nonzero eigenvalues $\mu_1,\ldots,\mu_r$ and the corresponding eigenvectors $v_1,\ldots,v_r$ of $({T_a}^{\dagger}T_b)^T$.\\
Reorder these eigenvectors and their eigenvalues  to make sure that the corresponding eigenvalues are reciprocal
 (i.e., $\lambda_i \mu_i=1$).\\
Solve for $w_i$ in the linear system $T=\sum_{i=1}^r u_i \otimes v_i \otimes w_i$, output this decomposition.
\caption{decomposition by simultaneous diagonalization (sometimes called "Jennrich's algorithm").}
\end{algorithm}

%\caption{decomposition by simultaneous diagonalization (sometimes called "Jennrich's algorithm") for tensors of rank $n$.}
%\end{figure}
An analysis of %Jennrich's 
this algorithm can be found in~\cite[Section 3.3]{moitra18}, where it is called ``Jennrich's algorithm.'' In particular, it can be shown that with high probability over the choice of the coefficients $a_k,b_k$,  each of the matrices at steps 2 and 3 have exactly $r$ distinct nonzero eigenvalues.
 Moreover, these eigenvalues are reciprocal (refer to step 4). An optimized version of %Jennrich's decomposition 
the simultaneous diagonalization algorithm (and a detailed complexity analysis) for the special case of 
symmetric tensors can be found in~\cite{KS23b,KS24}.

\section{Jennrich's uniqueness theorem from rank arguments}

Nowadays, the best known proof of the Jennrich uniqueness is probably the spectral one. It has the advantage of yielding an  efficient tensor decomposition algorithm based on simultaneous diagonalization (see~\cite[Chapter 3]{moitra18}, which emphasizes the algorithmic point of view).
This uniqueness theorem can be traced back (in a slightly less general form)   to Harshman~\cite{harshman70}, where it is attributed to Jennrich. 
Another version of the uniqueness theorem appears in a second paper by Harshman~\cite{harshman72}, and the proof seems closer to a simultaneous diagonalization argument. 
Jennrich's uniqueness theorem also follows from the (more involved) Kruskal uniqueness theorem~\cite{kruskal77,rhodes10}.
In this appendix we give a proof which builds on the rank arguments from Section~\ref{sec:rank}.

\begin{lemma} \label{lem:span}
If the vectors $w_1,\ldots,w_q$ in~(\ref{eq:matrixvector}) are linearly independent, the span of the matrices $Z_1,\ldots,Z_p$ is equal to the span of $M_1,\ldots,M_q$.
\end{lemma}
\begin{proof}
It follows immediately from~(\ref{eq:slicematrix}) that the span of the $Z_k$ is included in the span of the $M_{\ell}$, and for this no hypothesis on the $w_{\ell}$ is needed. 

For the converse first observe that the entries of $T$ are given by the formula: 
\begin{equation} \label{eq:entries}
T_{ijk} = \sum_{\ell=1}^q (M_{\ell})_{ij} w_{\ell k}.
\end{equation}
Let us denote by $t_{ij}$ and $m_{ij}$ the column vectors of size $p$ and $q$ with respective entries $(T_{ijk})_{1 \leq k \leq p}$ and $((M_{\ell})_{ij})_{1 \leq  \ell \leq  q}$. 
We can rewrite~(\ref{eq:entries}) as the matrix-vector product $t_{ij}=Wm_{ij}$ where $W$ has $w_1,\ldots,w_q$ as column vectors. 
Since these vectors are linearly independent, there is a $q \times p$ matrix $W'$ such that $W'W=I_q$. From this we obtain $m_{ij}=W' t_{ij}$ and the identity
$$M_{\ell} = \sum_{k=1}^p w'_{k \ell}Z_k,$$
which is  converse to~(\ref{eq:slicematrix}).
\end{proof}
%Remark:  the $w_i$ are assumed to be linearly independent in~\cite{johnston23}.

\begin{lemma} \label{lem:rank2b}
Let $M \in M_{2m,n}(K)$ be a matrix of the form 
$$M=\begin{pmatrix}
u v^T \\ u'v'^T
\end{pmatrix}$$
where $u,u' \in K^m$ and $v,v' \in K^n$. If the two vectors vectors $u,u'$ are linearly independendent, and  the two vectors vectors $v,v''$ are also linearly independendent, then $\rk M = 2$.
\end{lemma}
\begin{proof}
In order to show that $\rk M=2$, we'll show that $\dim \ker M = n-2$.
A vector $x \in K^n$ belongs to $\ker M$ if and only if $(v^Tx)u+(v'Tx)u'=0$. Since $u,u'$ are linearly independent, this is equivalent to $v^Tx=v'Tx=0$. Since $v,v'$  are linearly independent, the space of solutions of this linear system has dimension $n-2$.
\end{proof}

\begin{theorem}[Jennrich's uniqueness theorem] \label{th:jennrich2}
Let $T=\sum_{i=1}^r u_i \otimes v_i \otimes w_i$ be a tensor of format $m \times n \times p$ such that:
\begin{itemize}
\item[(i)] The vectors $u_i$ are linearly independent.
\item[(ii)] Every pair of vectors in the set $\{v_i;\ 1\leq i \leq r\}$ is linearly independent.
%The vectors $v_i$ are pairwise linearly independent.
\item[(iii)] The vectors $w_i$ are linearly independent.
\end{itemize}
Then $\rk(T)=r$, and the decomposition of $T$ as a sum of $r$ rank one tensors is essentially unique.
\end{theorem}
Compared to the usual statement of this uniqueness theorem, we have switched the roles of the $v_i$ and $w_i$: it is usually assumed that the $w_i$ (instead of the $v_i$) are pairwise linearly independent (see Theorem~\ref{th:jennrich}).
The proof below hinges on the fact that $u_1 \otimes v_1,\ldots,u_r \otimes v_r$ are up to scaling the only matrices of rank 1 in
the span of the 3-slices of $T$. With the usual statement of the uniqueness theorem, one would have to work with the 1-slices or the 2-slices instead of the 3-slices.
\begin{proof}
Let $V$ be the span of the 3-slices of $T$. By Lemma~\ref{lem:span}, $V = \text{Span}(u_1 v_1^T,\ldots,u_r v_r^T)$.
Since the $u_i$ are linearly independent and the $v_i$ nonzero, these $r$ matrices are linearly independent.
In particular, $\dim V = r$ and $V$ contains $r$ matrices of rank 1. 
Furthermore, by Lemma~\ref{lem:rank2b} we can apply Theorem~\ref{th:uniquerank} to $V$ and to its basis 
$M_1=u_1v_1^T,\ldots,M_r=u_rv_r^T$. As a result, taking $i=1$ in~(\ref{eq:smallest}), we conclude that $M_1,\ldots,M_r$
are up to scaling the only rank-1 matrices in $V$.

Consider now any other decomposition $T=\sum_{i=1}^{r'} u'_i \otimes v'_i \otimes w'_i$. 
Note that $V \subseteq \text{Span}(u'_1 {v'_1}^T,\ldots,u'_{r'} {v'_{r'}}^T)$, but we have seen that $\dim V = r$.
Hence $r' \geq r$, and we have shown that $\rk(T)=r$.

For the remainder of the proof we will assume that $r'=r$; it remains to show that the rank-1 tensors $u_i \otimes v_i \otimes w_i$ are up to permutation the same as $u'_i \otimes v'_i \otimes w'_i$.
We have just seen that $V \subseteq \text{Span}(u'_1 {v'_1}^T,\ldots,u'_r {v'_r}^T)$. Since $\dim V = r$, this must be an equality: 
$V = \text{Span}(u'_1 {v'_1}^T,\ldots,u'_r {v'_r}^T)$, i.e., the matrices $u'_1 {v'_1}^T,\ldots,u'_r {v'_r}^T$ form a basis of $V$.
But we have seen that $u_1 v_1^T,\ldots,u_r v_r^T$ are up to scaling the only rank-1 matrices in $V$.
We conclude that $u_i v_i^T = u'_i {v'_i}^T,$ up to scaling and permutation.
Finally, we observe that there is a unique way of writing each 3-slice of $T$ as a linear combination of the  $u_i v_i^T$ since these matrices form a basis of $V$. This establishes the uniqueness of the vectors $w_1,\ldots,w_r$, and completes the proof.
\end{proof}

%% file: moore-penrose.tex
In this section we assume that $K$ is the field of real or complex numbers.
We recall the following properties of the Moore-Penrose inverse.
\begin{proposition} \label{prop:pseudo}
Consider two matrices $A \in M_{m,n}(K)$, $B \in M_{n,p}(K)$. Their Moore-Penrose inverses 
$A^{\dagger} \in M_{n,m}(K)$, $B^{\dagger} \in M_{p,n}(K)$ satisfy the following properties:
\begin{itemize}
\item[(i)] If $A$ has linearly independent columns, $A^{\dagger} A = I_n$.
\item[(ii)] If $B$ has linearly independent rows, $B B^{\dagger} = I_m$.
\item[(iii)] If $A$ has linearly independent columns or $B$ has linearly independent rows, then $(AB)^{\dagger} = B^{\dagger} A^{\dagger}$. 
\end{itemize}
\end{proposition}
We will switch back and forth between the point of views of Theorems~\ref{th:unique1} and~\ref{th:unique2}.
Recall that to go from the latter (the matrix-vector point of view) to the former (the viewpoint of ``ordinary rank decompositions''), we just need to write each matrix $M_{\ell}$ as the following sum of $\rk(M_{\ell})$ matrices of rank 1:
\begin{equation} \label{eq:rank1}
M_{\ell}=\sum_i u_i \otimes v_i.
\end{equation}
Note that $\Ima(M_{\ell})$ is the span of the vectors $u_i$ occurring in~(\ref{eq:rank1}).
As in Section~\ref{sec:proof}, for any $c \in K^p$ we denote by $T_c$ the linear combination of slices  $\sum_{k=1}^p c_kZ_k$, and $D_c$ denotes the linear combination $\sum_{k=1}^p c_kD_k$ in~(\ref{eq:diag}). By~(\ref{eq:slice}), $T_c = UD_c V^T$.
The computations that follow are reminiscent of the classical treatment of Jennrich's algorithm as in, e.g.,~\cite{moitra18} (see also~\cite{leurgans93}).
\begin{lemma} \label{lem:pseudo}
For any $a  \in K^p$ and a generically chosen $b \in K^p$, $T_a T_b^{\dagger}=U D_a D_b^{-1} U^{\dagger}$.
\end{lemma}
\begin{proof}
For a generically chosen $b \in K^p$, $\rk(D_b)=r$ as shown in the proof of Proposition~\ref{prop:invert1}. This implies that $D_bV^T$ has linearly independent rows. Since $U$ has linearly independent columns, %a well-known property of the Moore-Penrose inverse implies that 
Proposition~\ref{prop:pseudo}.(iii) implies
$T_b^{\dagger}=(D_b V^T)^{\dagger} U^{\dagger}$. By the same token, since $D_b$ has linearly independent columns and $V^T$ linearly independent rows, $(D_b V^T)^{\dagger} = (V^T)^{\dagger} D_b^{\dagger}=(V^T)^{\dagger} D_b^{-1}$.

To conclude, we multiply by $T_a=UD_a V^T$ and use Proposition~\ref{prop:pseudo}.(ii): %another well-known property of the Moore-Penrose inverse: 
$V^T (V^T)^{\dagger} = I_r$ since the $r$ rows of $V$ are linearly independent.
\end{proof}

\begin{proposition} \label{prop:image}
For  generically chosen $a,b \in K^p$, $T_a T_b^{\dagger}$ has exactly $q$ distinct nonzero eigenvalues and the corresponding eigenspaces are $\Ima(M_1),\ldots,\Ima(M_q)$.
\end{proposition}
\begin{proof}
We first check that the columns of $U$ (i.e., the $u_i$) are eigenvectors of $T_a T_b^{\dagger}$: by the previous lemma,
$T_a T_b^{\dagger}U=(U D_a D_b^{-1} U^{\dagger}) U = UD_a D_b^{-1}$. Here we use the fact that $U^{\dagger} U=I_r$
since the $r$ columns of $U$ are linearly independent  (Proposition~\ref{prop:pseudo}.(i)). Note that $u_i$ is associated to some eigenvalue
 $\lambda=\langle w_{\ell},a \rangle / \langle w_{\ell},b \rangle$ where $w_1,\ldots,w_q$ are the $q$ distinct ``third mode vectors'' occurring in a decomposition of $T$. 
This eigenvalue is nonzero for a generic $a \in K^p$.
 Moreover,  by~(\ref{eq:rank1}),
$\Ima M_{\ell}$ is included in the corrresponding eigenspace  $V_{\lambda}$. %{\bf TBD: these eigenvalues are distinct}
For generically chosen $a,b \in K^p$, these $q$ eigenvalues are distinct. Like in the analysis of the standard version of Jennrich's algorithm, this follows from the fact that the $w_{\ell}$ are pairwise linearly independent (for a detailed argument, see the proof of Lemma~\ref{lem:image2}.(ii)).

In order to complete the proof of the proposition, we still need to derive the converse inclusion
 ($V_{\lambda} \subseteq \Ima M_{\ell}$) and we need to show that 0 is the only possible other eigenvalue.
 We will in fact show that 0 is an eigenvalue of multiplicity $m-r$, which achieves these two goals at once.
 For this, recall that $\rk U^{\dagger} = \rk U = r$, hence $\dim \ker U^{\dagger} = m-r$. 
 Moreover, $ \ker U^{\dagger} \subseteq \ker T_a T_b^{\dagger}$ by Lemma~\ref{lem:pseudo}.
 This shows that 0 has multiplicity at least $m-r$ as an eigenvalue. This is in fact the exact value of the multiplicity since we have already found other eigenvalues (the $\langle w_{\ell},a \rangle / \langle w_{\ell},b \rangle$) whose multiplicities sum at least to $r$.
\end{proof}

It is possible to eliminate the Moore-Penrose inverse from the above computations (and also from the ordinary Jennrich algorithm),
and to replace it by the ordinary matrix inverse. We briefly sketch how to do so. First we can obtain the span of the $u_i$
as $\Ima(T_a)$ for a generic $a \in K^p$. We can likewise obtain the span of the $v_i$ as $\Ima(T_a^T)$.
Then we can perform a change of basis to reduce to the case of a tensor $T'$ of format $r \times r \times p$ and rank $r$.
For such a tensor, the Moore-Penrose inverse can be replaced by the ordinary inverse. For the symmetric version of Jennrich's algorithm, this approach is worked out in detail (with numerical error bounds) in~\cite{KS24}, see in particular Section~1.3.1 of that paper.

The modifications above result in an algorithm that is applicable to an arbitrary field (not just to the real and complex numbers).
We have already presented such an algorithm in Section~\ref{sec:images2}. One difference is that the algorithm of Section~\ref{sec:images2} computes the spaces $\Ima(M_{\ell})$ ``in one go'' whereas the algorithm that we have just sketched first computes the span of the $u_i$ (i.e., the direct sum of the $\Ima(M_{\ell})$) before computing the spaces $\Ima(M_{\ell})$
themselves in a second stage.